\documentclass[english,aps,prl,amsmath,amssymb,twocolumn,letterpaper,citeautoscript,superscriptaddress,longbibliography]{revtex4-1}

\usepackage{mathtools} 

\usepackage[normalem]{ulem} 
\usepackage{svg}
\usepackage{float}
\usepackage{comment}
\usepackage{amsmath}
\usepackage{physics}
\usepackage{graphicx}
\usepackage{subfigure}
\usepackage{babel}
\usepackage{amstext}
\usepackage{esint}
\usepackage{cases}

\usepackage{hyperref}
\hypersetup{bookmarks=false,linkcolor=blue,urlcolor=blue,colorlinks,citecolor=blue}

\usepackage{pdfpages}	
\usepackage{pgffor}		

\usepackage{accents}

\usepackage{amssymb}
\newcommand{\ols}[1]{\mskip.5\thinmuskip\overline{\mskip-.5\thinmuskip {#1} \mskip-.5\thinmuskip}\mskip.5\thinmuskip} 
\newcommand{\olsi}[1]{\,\overline{\!{#1}}} 

\makeatletter
\AtBeginDocument{\let\LS@rot\@undefined} 
\makeatother							 

\usepackage{xcolor}

\usepackage{amsthm}

\newtheorem{theorem}{Theorem}

\begin{document}

\title{Characterizing maximally many-body entangled fermionic states by using \texorpdfstring{$M$}{M}-body density matrix}

\author{Irakli Giorgadze}

\affiliation{Department of Physics and Astronomy, Purdue University, West Lafayette, Indiana 47907 USA}

\author{Haixuan Huang}

\affiliation{Department of Physics and Astronomy, Purdue University, West Lafayette, Indiana 47907 USA}

\author{Jordan Gaines}

\affiliation{Department of Physics and Astronomy, Purdue University, West Lafayette, Indiana 47907 USA}

\author{Elio J. K{\"o}nig}
\affiliation{Department of Physics, University of Wisconsin-Madison, Madison, Wisconsin 53706, USA}
\affiliation{
Max-Planck Institute for Solid State Research, 70569 Stuttgart, Germany
}

\author{Jukka I. V\"ayrynen}

\affiliation{Department of Physics and Astronomy, Purdue University, West Lafayette, Indiana 47907 USA}

\date{\today}
\begin{abstract}

Fermionic Hamiltonians play a critical role in quantum chemistry, one of the most promising use cases for near-term quantum computers. However, since encoding nonlocal fermionic statistics using conventional qubits results in significant computational overhead, fermionic quantum hardware, such as fermion atom arrays, were proposed as a more efficient platform.
In this context, we here study the many-body entanglement structure of fermionic $N$-particle states by concentrating on $M$-body reduced density matrices (DMs) across various bipartitions in Fock space. 
The von Neumann entropy of the reduced DM is a basis independent entanglement measure which generalizes the traditional quantum chemistry concept of the one-particle DM entanglement, which characterizes how a single fermion is entangled with the rest. 
We carefully examine  
upper bounds on the $M$-body entanglement, which are analogous to the volume law of conventional entanglement measures. To this end we establish a connection between $M$-body reduced DM and the mathematical structure of hypergraphs. Specifically, we show that a special class of hypergraphs, known as $t$-designs, corresponds to maximally entangled fermionic states. 
Finally, we explore fermionic many-body entanglement in random states. 
We semianalytically demonstrate that the distribution of reduced DMs associated with random fermionic states corresponds to the trace-fixed Wishart-Laguerre  random matrix ensemble. In the limit of large single-particle dimension $D$ and a non-zero filling fraction, random states asymptotically become absolutely maximally entangled.

\end{abstract}
\maketitle

\section{Introduction}\label{sec:introduction}

Quantum chemistry simulation is among the best use-cases for a near-term quantum computer~\cite{DaleyZoller2022,ReiherTroyer2017}. 
A substantial fraction of quantum chemistry applications for quantum computing require studying fermionic Hamiltonians. 
Because of the fermionic exchange sign, qubits are not naturally suited for simulating fermionic systems and come with an overhead~\cite{PhysRevA.92.062318, AbramsLloyd1997, OrtizGubernatisKnillLaflamme2002, BravyiKitaev2002, WhitfieldBiamonteAspuruGuzik2011, Ball2005, VerstraeteCirac2005, WhitfieldHavlicekTroyer2016} that may limit near-term applications.
Thus, it is more natural to use natively fermionic quantum hardware such as ultacold atoms~\cite{GonzalezCuadraZoller2023}.  
However, little emphasis has been aimed at fermions in the study of quantum information-theoretic concepts such as entanglement entropy~\cite{NielsenChuang2010}.

Conventional entanglement entropy, defined by using an entanglement cut in a single-particle basis such as in real space, does not adequately characterize the complexity of fermionic many-body states.  For example, the ground state of a quadratic (i.e., a non-interacting) fermionic Hamiltonian can be obtained using classical computers in polynomial time and can be represented as a single Slater determinant~\cite{knill2001,Terhal2002,Bravyi2004}. 
(We assume throughout a conserved number of fermions~\cite{BCS}.) 
Despite its simplicity, this Slater determinant ground state can be highly entangled when the conventional entanglement measure is used~\cite{Wolf2006, GioevKlich2006, BarthelChungSchollwock2006, LiDingYuRoscildeHaas2006, Swingle2010}. 
As a concrete example, consider the electronic ground state of a simple metal. This state is a single Slater determinant in momentum space (the filled Fermi sea) and hence is not a complex quantum state.  Despite this fact, standard entanglement measures such as the entanglement entropy associated with bipartitions in real space attribute a strong entanglement to this rather elementary state~\cite{Swingle2010}. 
This entanglement arises from the antisymmetrization of the fermionic wave function which leads to ``trivial'' correlations for fermionic states \cite{GiulianiVignale2005}. 
Moreover, the conventional entanglement measure depends on the single-particle basis. 
In contrast, traditional entanglement measures applied to qubits do not change under single-qubit rotations~\cite{NielsenChuang2010}.

In quantum chemistry, the intrinsic complexity of electron-electron interactions generally prevents ground states from being expressed as a single Slater determinant, an issue that falls under the umbrella of ``correlation effects''~\cite{PhysRev.46.1002, SzaboOstlund1996, Jensen2017}. 
Modern quantum materials provide a diverse array of examples in which these correlation effects become fundamental. Systems such as quantum spin liquids~\cite{SavaryBalents2016}, non-Fermi liquids~\cite{PhillipsAbbamonte2022}, and fractional quantum (anomalous) Hall states~\cite{ParkXu2023} exemplify the unique phases that emerge from complex many-body interactions.  These phases differ from those of classic solid-state theories in the distinctive nature of their many-body entanglement.  For instance, the fractional quantum Hall Laughlin state~\cite{Laughlin1983, ZengXu2002, HaqueSchoutens2007, DubailRezayi2012} represents a highly entangled quantum system in which correlation effects are manifestly non-trivial. Unlike traditional solid-state phases, these states exhibit exotic properties such as fractionalized excitations and topologically protected edge modes. The complexity of these strongly correlated systems makes their simulation on classical computers especially challenging, corresponding to NP-hard problems~\cite{HastingsO'Donnell2022}.

The number of Slater determinants required to represent a quantum state depends on the choice of single-particle basis. 
To analyze the correlation effects in many-body systems, we adopt a basis-independent approach that involves the use of a single-particle density matrix (DM) and its $M$-body generalization~\cite{GigenaRossignoli2021}, widely used in early quantum many-body theory~\cite{cioslowski2012,CarlsonKeller1961,Coleman1963,Ando1963,Sasaki1965} prior to the widespread adoption of field-theoretical techniques~\cite{abrikosov2012methods}. 
Since the $M$-body DM is a central object also in the present article, we will review it next. 

For a system of fermions with a $D$-dimensional single-fermion Hilbert space, a general $N$-fermion quantum state $|\Psi\rangle$ can be expressed as a linear combination of $N$-particle Slater determinant (SD) states which span the $N$-particle sector of Fock space. The $M$-body DM associated with $|\Psi\rangle$ is defined as
\begin{equation}\label{eq:definition of M-body DM}
    (\rho_{\Psi}^{(M)})_{\alpha \beta} = \langle\Psi|C^{(M)\dagger}_{\beta} C^{(M)}_{\alpha} |\Psi\rangle,
\end{equation}
where $C^{(M)\dagger}_{\alpha} = c^{\dagger}_{i_1} \cdots c^{\dagger}_{i_M}$ is an $M$-fermion creation operator and $c^{\dagger}_i$ creates a single fermion in the single-particle state $i=1,\dots,D$. The collective indices $\alpha, \beta$ label all $\binom{D}{M}$ ordered combinations of the form $[i_1, \dots ,i_M]$ such that $i_1 < i_2 < \cdots < i_M$.  Thus, $\rho^{(M)}$ is a $\binom{D}{M} \times \binom{D}{M}$ matrix. (For brevity, we will often leave out the subscript of $\rho^{(M)}$ when discussing generic states.)
When $M=1$, this expression reduces to the familiar single-particle density matrix, while higher values of $M$ can be used to capture more complex entanglement structures that characterize strongly correlated systems.  In this framework, the entanglement entropy can be interpreted with respect to an entanglement cut in Fock space, partitioning the system according to the number of particles in each sector, see Fig.~\ref{fig:summary Fock space}.  Importantly, the eigenvalues of $M$-body DM remain invariant under a change of the single-particle basis: $c^{\dagger}_{i} = \sum_k u_{ik}\tilde{c}^{\dagger}_k$ where $u$ is a unitary matrix and $\tilde{c}^{\dagger}_k$ is the fermionic creation operator in the new basis~\cite{GigenaRossignoli2021}.  Hence, the $M$-body DM allows us to characterize quantum correlations between subsets of particles in basis-independent way. 
(Other measures of fermion state complexity have been proposed recently, see e.g. Ref.~\cite{VanhalaOjanen2024}.)

In this paper, we specifically address upper bounds on the multi-particle entanglement~\cite{ZozulyaSchoutens2008}. Loosely speaking, these are analogous to the volume law of conventional entanglement measures.  By using a mapping to \textit{hypergraphs}, we
determine when these bounds can be saturated Fig.~\ref{fig:SummaryFig}(a),(c). To study maximally entangled states, we focus on small number of fermions, $N$, and small single-particle Hilbert space dimension, $D$.  This choice is limited by numerical feasibility but is nevertheless relevant for near-term quantum computing applications. 
Additionally, our investigation includes a comprehensive study of randomly sampled $N$-fermion states to explore their multi-particle entanglement properties and get insights about their behavior in the large-$D$ limit. We demonstrate that statistics of $M$-body DMs are well-described by the \textit{trace-fixed Wishart-Laguerre} random matrix ensemble. Our findings suggest that the eigenvalue distributions of these DMs differ significantly depending on whether the entanglement upper bound can be saturated: Fig.~\ref{fig:summary Eigenvalues}.

Crucially, for systems with two-body interactions, the ground state energy is determined by
minimizing a functional of $\rho^{(2)}$ over the space of all possible two-body DMs, a space which contains the information about all one- and two-particle static correlation functions~\cite{Husimi1940,Lowdin1955, Mayer1955,Tredgold1957}. 
Moreover, linear response functions of one-body observables can be cast in terms of a two-body DM~\cite{GiulianiVignale2005}. 
Therefore, although a general $M$-body approach is discussed throughout this paper, our numerical results focus on one- and two-body DMs. A challenge in such optimization schemes lies in \textit{$N$-representability problem}, which involves determining whether a given two-body (or more generally, $M$-body) DM is derivable from a valid $N$-particle quantum state~\cite{Coleman1963, Tredgold1957,Coulson1960}. This problem is QMA-hard, highlighting the difficulty of ensuring that $\rho^{(2)}$ is compatible with an antisymmetrized $N$-particle wave function~\cite{LiuVerstraete2007,SchuchVerstraete2009}. $M$-body DM methods have recently regained attention through applications of bootstrapping methods in strongly correlated systems~\cite{GaoKhalaf2024}.

\begin{figure}[!htb]
    \centering
    \subfigure[]{
        \includegraphics[width=0.48\columnwidth]{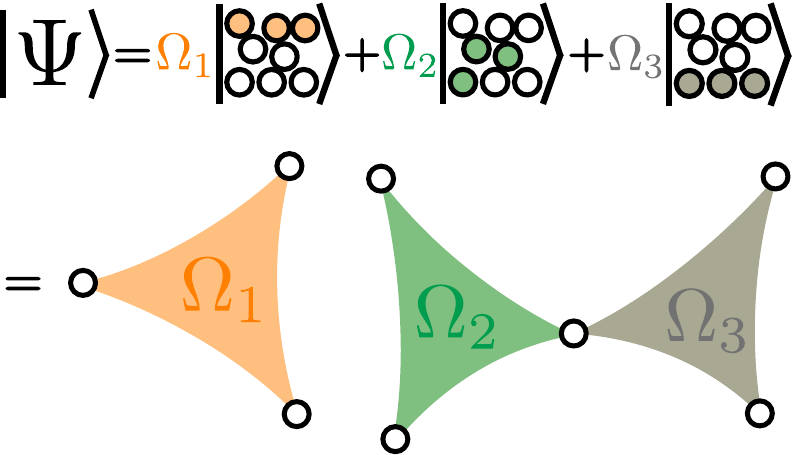}
        \label{fig:summary hypergraph}
    }
    \hspace{-1em}
    \subfigure[]{
        \includegraphics[width=0.48\columnwidth]{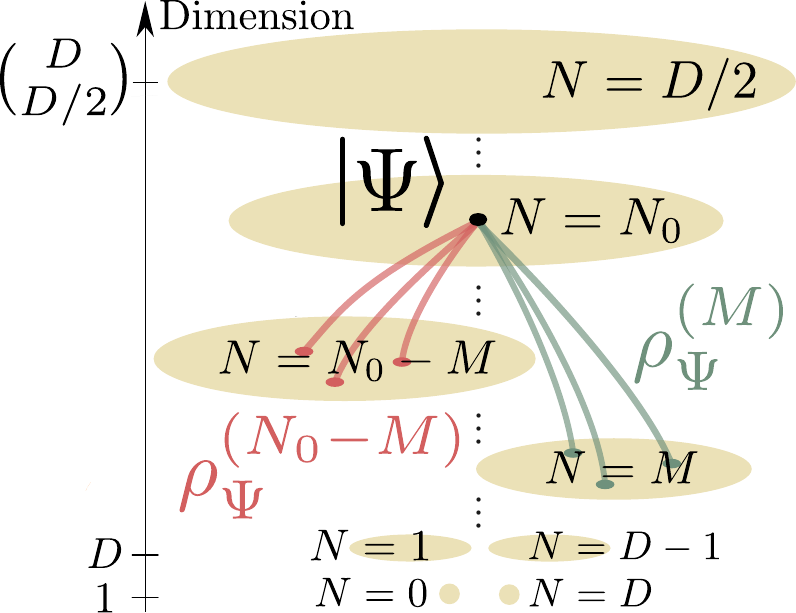}
        \label{fig:summary Fock space}
    }
    \\[-1.5em] 
    \subfigure[]{
        \includegraphics[width=0.48\columnwidth]{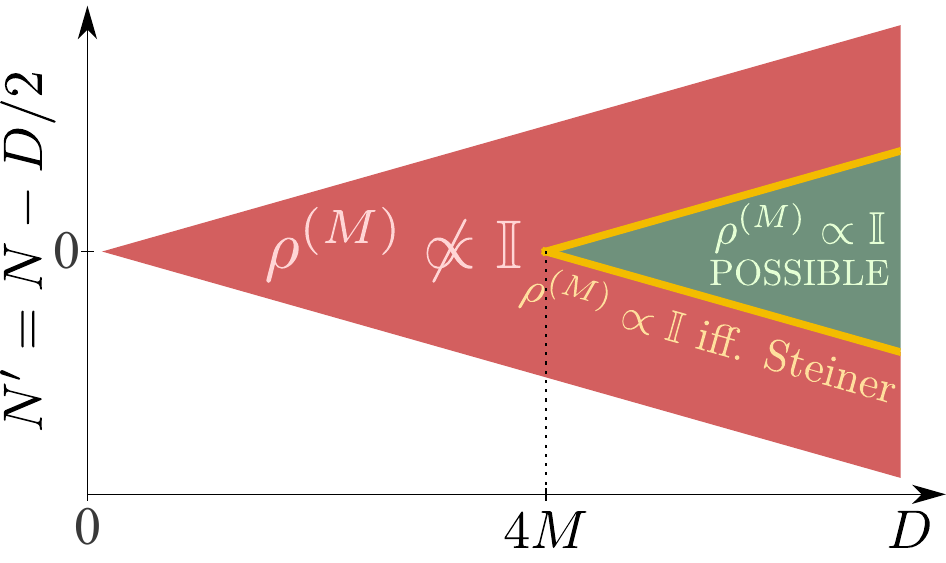}
        \label{fig:summary D-N plane}
    }
    \hspace{-1em}
    \subfigure[]{
        \includegraphics[width=0.48\columnwidth]{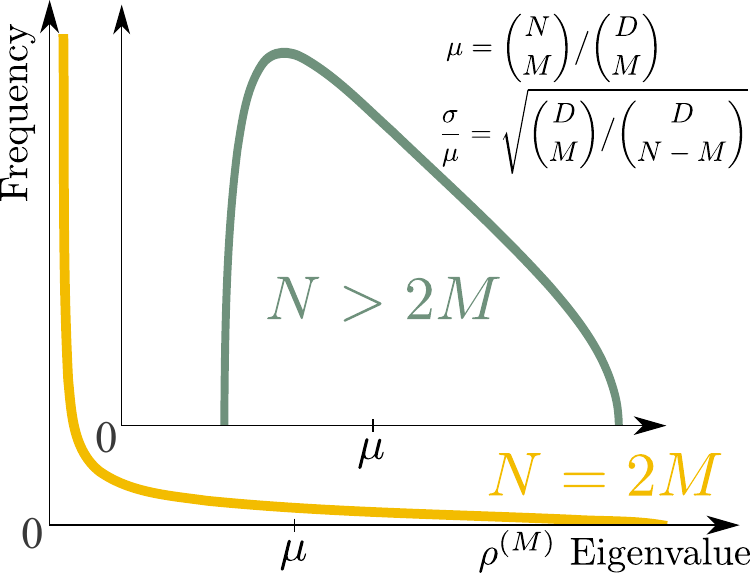}
        \label{fig:summary Eigenvalues}
    }
    \caption{
    Characterizing maximally many-body entangled fermionic states with $M$-body density matrix. (a)~An $N$-fermion pure state $\ket{\Psi}$ and its hypergraph representation. In a given basis of the $D$-dimensional single-particle Hilbert space, each $N$-fermion state can be represented by a complex-weighted $N$-uniform hypergraph of $D$ vertices. (b)~An $N_0$-fermion pure state $\ket{\Psi}$ and its many-body density matrices (DMs) located in the Fock space.  
    The DMs $\rho_\Psi^{(M)}$ and $\rho_\Psi^{(N_0-M)}$ characterize the
    correlation between $M$ fermions with the other $N_0-M$ fermions. 
    These DMs correspond to mixed states in the $M$-fermion and ($N_0-M$)-fermion subspaces, respectively. 
    A maximally mixed state is said to be maximally $M$-body entangled.  Such a state has $\rho_\Psi^{(M)}=\mu \mathbb{I}$, where $\mu=\binom{N}{M}/\binom{D}{M}$ is fixed by normalization. (c)~The existence of $N$-fermion states of single-particle dimension $D$ with $\rho^{(M)}\propto\mathbb{I}$ for various $(D,N)$ values. Such a state does not exist if $N<2M$ or $N>D-2M$ (red region). When $N=2M$ or $N=D-2M$, such a state exists if and only if its associated hypergraph is a Steiner system (yellow lines, see Sec.~\ref{subsec:existence of maximally entangled states for given D,N}). For $2M<N<D-2M$, such a state might or might not exist (green region). The plot is symmetric about $N=D/2$ due to particle-hole symmetry. (d)~Distribution of $\rho^{(M)}$ eigenvalues for a random $N$-fermion state in the large $D$ limit. When $N>2M$, the distribution is peaked with mean value $\mu$ and standard deviation $\sigma$, indicating that the random state is highly $M$-body entangled. 
    In the case when $N=2M$, however, only rare Steiner systems (see Sec.~\ref{subsec:entanglement of random state}) are maximally entangled states while a typical state is not. 
    }
    \label{fig:SummaryFig}
\end{figure}

\subsection{Qualitative summary of the main results}\label{subsec:non-technical summary}

We will now give a qualitative summary of the main results of our paper; see also Fig.~\ref{fig:SummaryFig}. 
Our primary object of study is the  $M$-body DM $\rho^{(M)}$, Eq.~(\ref{eq:definition of M-body DM}), and the associated $M$-body (von Neumann) entanglement entropy, 
\begin{equation}\label{eq:entropy1}
S( \rho^{(M)})=-\text{Tr}(\rho^{(M)} \ln{\rho^{(M)}}),  
\end{equation}
of a given $N$-fermion pure state $| \Psi \rangle$. 
We are particularly interested in states where $M$ fermions are maximally entangled with the remaining $N-M$ particles. 
The maximum entanglement entropy corresponds to the case when the $M$-body DM is maximally mixed, i.e.
$\rho^{(M)}\propto\mathbb{I}$. 
Therefore, to find a maximally entangled state, it is sufficient to find a state whose $M$-body DM is proportional to the identity matrix $\mathbb{I}$. 
The proportionality coefficient is fixed by the normalization $\text{Tr}\rho^{(M)} = \binom{N}{M}$.
This convention ensures that the entanglement entropy vanishes for SD states~\cite{GigenaRossignoli2021}. We therefore have the entropy bounds 
\begin{equation}
    0 \leq S(\rho^{(M)}) \leq  \binom{N}{M} \ln \left [\frac{ \binom{D}{M}}{ \binom{N}{M}  }\right].
\end{equation}
We mention that other normalization conventions can be found in the literature, see e.g. Ref.~\cite{HaqueSchoutens2009}. 

We will next give qualitative summaries of our results for maximally entangled states for small systems (small $N,D$,  relevant for near-term quantum devices) and then for random fermionic states. 

\subsubsection{Maximally entangled states}
 
The 1-body DM is insufficient for fully characterizing the entanglement structure of a fermionic state because it cannot account for all multi-particle correlations. 
Hence, the aim of the present paper is to provide a mathematical framework for understanding the properties of the $M$-body DM and the associated $M$-body entanglement entropy across various fermionic states.

As an illustrative example, we can consider the Greenberger–Horne–Zeilinger (GHZ) state~\cite{GreenbergerZeilinger2007} of $N$ fermions in $D=2N$ dimensional single-particle Hilbert space: $|\Psi_\text{GHZ}\rangle = \frac{1}{\sqrt{2}}(c^{\dagger}_1 \cdots c^{\dagger}_N + c^{\dagger}_{N+1} \cdots c^{\dagger}_{2N}) |0\rangle$. 
This state is a classic example of a highly-entangled state: it is an equal superposition of two half-filled Slater determinants with no joint occupied orbitals~\footnote{
Physically,  $|\Psi_\text{GHZ}\rangle$ may arise, for instance, in a model of 
spinless fermions described by 1D $t-V$ model~\cite{ZozulyaSchoutens2008}, where strong interaction  gives rise to a ground state that is in equal superposition of two charge density wave states. 
Naturally, such a state would decohere fast since a local measurement can distinguish the two components. }. 
Correspondingly, the occupation numbers 
$(\rho^{(1)}_\text{GHZ})_{ii} = 1/2$ for all $i=1,\dots,D$ as can be verified by using Eq.~(\ref{eq:definition of M-body DM}); therefore, $\rho^{(1)}_\text{GHZ}\propto\mathbb{I}$ is maximally mixed. 
Indeed, for the 1-body DM $\rho_\text{GHZ}^{(1)}$, the entanglement entropy is $S(\rho_\text{GHZ}^{(1)}) = N\ln{2}$, which corresponds to maximum entropy value for 1-body DM $S^{(1)}_\text{max}$ (see Eq.~(\ref{eq:S maximum entropy general}), below). 

However, the GHZ state is missing some two-body correlations. 
For instance, the \textit{pair} of orbitals $N,N+1$ is not occupied in either of the terms in $|\Psi_\text{GHZ}\rangle$ and thus $(\rho^{(2)}_\text{GHZ})_{N\,N+1;N\,N+1}=0$. Therefore, $\rho^{(2)}_\text{GHZ}$ is not proportional to the identity and not maximally mixed. 

To estimate the $S(\rho_\text{GHZ}^{(2)})$ entanglement entropy, we need to count the unique pairs of orbitals present in the state. Since non-overlapping $N$ orbitals are occupied in each SD of the GHZ state, there are  $2\binom{N}{2} \approx N^2$ unique occupied pairs for large $N$. In contrast, for a 2-body DM proportional to the identity, the total number of pairs is  $\binom{D}{2} \approx D^2/2=2N^2$. Hence, approximately half of the orbital pairs are missing in the GHZ state. As a result, while the GHZ state is maximally entangled at the 1-body level, the 2-body DM exhibits only half of the maximum possible entropy. 
Indeed, for the GHZ state we find $S(\rho_\text{GHZ}^{(2)}) = \binom{N}{2}\ln{2}$, while the maximum entropy in this case is given by $S^{(2)}_\text{max} = \binom{N}{2} \ln{\left[\binom{D}{2}/\binom{N}{2}\right]}$. 
In the limit $N \gg 1$, this simplifies to $S^{(2)}_\text{max} \approx 2\binom{N}{2} \ln{2} = 2 S(\rho_\text{GHZ}^{(2)})$ in agreement with the above argument. As a result, while the GHZ state is maximally entangled at the 1-body level, the 2-body DM exhibits only half of the maximum possible entropy.

The above example raises multiple natural questions: for what values of $D$ and $N$ can we find a state whose 2-body (or more generally, $M$-body) DM is proportional to the identity? 
Does such a state even exist? 
As noted towards the end of the introduction, these  questions are hard due to the so-called $N$-representability problem. While one can always construct an $M$-body DM of the correct dimensions that is proportional to the identity, as shown schematically in Fig.~\ref{fig:summary Fock space}, it need not necessarily corresponds to a physical  $N$-fermion state. Our analysis reveals that not all single-particle dimensions $D$ and number of fermions $N$ support maximally $M$-body entangled states, see  Fig.~\ref{fig:summary D-N plane}.

The condition $\rho_{\Psi}^{(M)} \propto \mathbb{I}$ imposes two competing requirements for fermionic states. First, as illustrated by the GHZ example, the state must include all possible sets of $M$-orbitals to avoid vanishing diagonal elements in $\rho^{(M)}$. Second, the state must ensure that all off-diagonal elements of $\rho^{(M)}$ are zero. From Eq.~(\ref{eq:definition of M-body DM}), it is evident that any two SDs in the state must share strictly fewer than $N-M$ orbitals; otherwise, $\rho^{(M)}$ will exhibit non-zero off-diagonal element. Hence, if the state contains too few SDs, some $M$-orbital combinations will be missing, leading to vanishing diagonal elements. Conversely, if the state includes too many SDs, off-diagonal elements will be present. Thus, for a given $D$ and $N$, identifying the existence of such states generally requires exploring every possible set of SDs in which any pair of elements shares fewer than $N - M$ single-particle orbitals. This naturally limits the numerical search to small $D,N$ values.

Additionally, we need to ensure that the diagonal elements not only are all non-zero but also equal to each other. 
To get such uniform diagonal elements of $\rho^{(M)}$, the expansion coefficients must ``compensate'' for the frequency with which each $M$-orbital combination appears in the state. For instance, if a particular set of $M$-orbitals is present in three SDs while another set appears in only one, and all expansion coefficients are equal, the corresponding diagonal elements of $\rho^{(M)}$ would differ by a factor of three. To counteract this ``imbalance'' and equalize the diagonal elements, the expansion coefficients must appropriately ``weight'' the presence of each $M$-orbital combination, ensuring that all combinations are ``equally represented'' in the state in terms of their overall contribution to the diagonal elements of $\rho^{(M)}$. This motivates the application of \textit{hypergraph} theory and the  theory of special class of hypergraphs known as \textit{balanced block designs} or $t$-\textit{designs}, which naturally capture the idea of ``equal representation'' for each $M$-orbital combination. By leveraging hypergraph theory, we can significantly reduce the numerical search process while gaining valuable new insights. To proceed, we first establish the connection between fermionic states and hypergraphs (for more details, see Sec.~\ref{subsec:mini-review on hypergraphs}).

An $N$-fermion pure state in a set of $D$ single-particle orbitals can be represented by a hypergraph of $D$ vertices and some number of edges, see Fig.~\ref{fig:summary hypergraph}. Each edge is a collection of $N$ vertices, thus representing a single Slater determinant (and the orbitals occupied in it). The number of edges in the hypergraph representation of a state is therefore equal to the number of Slater determinants in that state. 
(Since each edge contains the same number $N$ vertices, the hypergraph is said to be ``$N$-uniform''; it is also undirected.)  
Moreover, by assigning complex weights to the edges, we can fully represent a specific $N$-fermion state $|\Psi\rangle$, see Fig. \ref{fig:summary hypergraph}.

Next, we will outline the type of hypergraphs that correspond to a fermionic state where each $M$-orbital combination is equally likely. 
An $N$-uniform hypergraph is known as a $t$-\textit{design} if every size $t$ subset of the vertex set appears in the same number of edges. In the context of fermionic states, this is equivalent to ensuring that every $t$-orbital combination is represented in the same number of SDs. This condition corresponds to the extreme case of ``equal representation'' of all $t$-orbital combinations when all expansion coefficients are equal in magnitude.  
For example, the $\frac{N}{2}$-body DM (equal bipartitioning) is maximally mixed only when the underlying fermionic state is an $\frac{N}{2}$-design with every $\frac{N}{2}$-orbital combination appearing in exactly one SD, see Fig.~\ref{fig:summary D-N plane}. The $t$-designs with any $t$-orbital combination appearing in exactly one SD are called \textit{Steiner systems}~\footnote{Some known Steiner systems provide examples of maximally entangled fermionic states with large $D,N$, such as the $2$-design with $D=1641, N=41, \lambda=1$ provided in Ref.~\cite{ColbournDinitz2006}, along with many other examples of $2$-designs. An example of a $3$-design Steiner system with $D=22, N=6, \lambda=1$ can be found in Ref.~\cite{Tonchev1986}, which is one of the derived designs from \textit{Witt system}}.

Not all maximally entangled fermionic states correspond to $t$-designs. More generally, we will see that the condition $\rho_{\Psi}^{(M)} \propto \mathbb{I}$ leads to a system of equations whose solutions correspond to specific fermionic states that are maximally $M$-body entangled. 
A given set of SDs that forms a state $| \Psi \rangle$ has a corresponding system of equations. Exploring all such systems numerically is possible for not too large $D,N$.  
Moreover, because $\rho_{\Psi}^{(M)}$ eigenvalues are invariant under single-particle basis changes, this system exhibits redundancy allowing us to restrict the set of SDs to be explored. 
In the language of hypergraphs, this redundancy translates to considering only the \textit{isomorphism classes} of hypergraphs when determining if system of equations can be solved. (Two hypergraphs are isomorphic if one can be transformed into the other by relabeling vertices and edges.) 
This insight plays a crucial role in simplifying the search process during numerical analysis. By employing established methods for determining hypergraph isomorphism \cite{McKayPiperno2014}, the number of systems of equations that need to be solved is significantly reduced. However, identifying isomorphism classes is itself computationally limited to relatively small values of $D$ and $N$. 
By using this method, we identify several  $D,N$ where we find  maximally entangled states that are not given by a $t$-design. 
The example of $D=18, N=8$ illustrates a case where no known $t$-design exists, yet we find numerically a fermionic state that gives $\rho^{(2)}\propto\mathbb{I}$ 
[see Fig.~\ref{fig:D and N mesh plot for M=1 and M=2}]. The explicit forms of such general states are, however, too complicated to present here.  

Finally, not all $D,N$ pairs admit a maximally entangled state. For instance, for $D=10, N=5$, a search through all isomorphism classes shows that $\rho^{(2)} \not\propto\mathbb{I}$ for any fermionic state. Since the $M$- and $N-M$-body entanglement entropies are identical, it also implies that no 3-body maximum entropy state exists for this $D,N$ pair.

\subsubsection{Random states}

The behavior of the $M$-body DM, $\rho^{(M)}$, and the associated von Neumann entropy, $S(\rho^{(M)})$, in the limit of large single-particle dimension $D$ is another major focus of the present work. As numerical simulations become increasingly impractical at higher single-particle dimensions, we shift our focus to the statistical behavior of these quantities in the large-$D$ limit, specifically considering random fermionic states. Using randomly generated fermionic states, we numerically obtained random physical $M$-body DMs and investigated their statistical properties. We find that the eigenvalues of $\rho^{(M)}$ for large $D$ (but already at $D < 100$) are narrowly distributed around the eigenvalue corresponding to maximally mixed $\rho^{(M)}$. Remarkably, the eigenvalue distribution matches the predictions derived from the trace-fixed Wishart-Laguerre (WL) ensemble (see Figs.~\ref{fig:SC distribution and std deviation as a function of D}--\ref{fig:TWL general distribution for N=4,5 M=2 D=20}). 

The WL ensemble represents an ensemble of square $n \times n$ Wishart matrices, $W = H H^{\dagger}$, where $H$ is an $n \times m$ complex matrix $(m \geq n)$ with independently and identically distributed Gaussian entries~\cite{PottersBouchaud2020}. We will see that in our case $n=\binom{D}{M}$ (which is also the dimension of $\rho^{(M)}$) and $m=\binom{D}{N-M}$ (which is the dimension of the space to be traced over to obtain $\rho^{(M)}$). In the large-$n,m$ limit (corresponding to $D \gg 1$), the eigenvalue distribution of $W$ converges to the well-known Mar\v{c}enko-Pastur distribution~\cite{PottersBouchaud2020}, whose shape is controlled by the single  parameter $c = n/m$.

In order to relate to $M$-body DMs, we consider the trace-fixed WL ensemble as a representation of the ensemble of $M$-body DMs.
In the limit $c \rightarrow 0$ (achieved when $D \gg N > 2M$), trace-fixed WL ensemble becomes equivalent to the \textit{Gaussian special unitary ensemble} (GSUE). We show that the distribution of eigenvalues of $\rho^{(M)}$ approaches the semi-circular distribution, which is the large-$D$ limit of GSUE. Both the mean $\mu$ and standard deviation $\sigma$ of this distribution decrease with increasing $D$. However, we will see that the ratio $\frac{\sigma}{\mu}=\sqrt{c}$ approaches zero as $c \to 0$. 
As a result, the eigenvalue distribution becomes highly peaked about the mean $\mu$, which represents the eigenvalue of the maximally mixed $\rho^{(M)}$. Correspondingly, as $c \rightarrow 0$, the mean of the entropy distribution rapidly converges to maximum $M$-body entanglement entropy value for the given dimension. Since this holds for arbitrary $M<N/2$, random states in the $c \to 0$ limit achieve maximum entanglement entropy for any bipartition. Such states are known as \textit{absolutely maximally entangled} states.

Allowing $c$ to be arbitrary but keeping the large $D$ limit (this corresponds to a non-vanishing filling $N/D$), the eigenvalue distribution of trace-fixed WL matrices converges to the rescaled Mar\v{c}enko-Pastur distribution. 
An interesting case is $M=N/2$ (equal bipartitioning) because $n=m$ and therefore $c=1$ regardless of $D$. Consequently, in the $D \gg N$ limit and when $N$ is odd, the eigenvalues of any $M$-body DM result in a semi-circular distribution. However, when $N$ is even, the eigenvalues of $\rho^{(N/2)}$ follow a different distribution shown in Fig.~\ref{fig:summary Eigenvalues}. This can be qualitatively explained by the fact that $\rho^{(N/2)}$ gives maximum entropy only when the underlying state is related to a Steiner system (see Sec.~\ref{subsec:mini-review on hypergraphs}), as demonstrated in Fig. \ref{fig:summary D-N plane}. Finally, we explicitly compute the mean value of the entanglement entropy for the case $c = 1$ (equal bipartitioning) and show that it stays lower than the maximum entropy by a fixed constant, independent of $D$. This is expected, as the contribution to maximum entropy in this case comes solely from the Steiner system.  
The numerical results are compared with analytical expressions, showing a high level of agreement (see Figs.~\ref{fig:SC distribution and std deviation as a function of D}--\ref{fig: mean entropy vs D for N=4 and M=1, M=2} in Sec.~\ref{sec:random states}).

\subsection{Overview}\label{subsec:overview}

The introductory Section~\ref{sec:characterization of DM and entropy} provides the review of key concepts focusing on the characterization of the $M$-body DM $\rho^{(M)}$ and the von Neumann entropy $S(\rho^{(M)})$.  To provide concrete examples, we apply these concepts to specific fermionic states. We explain the general form of the $N$-fermion state $|\Psi\rangle$ and its expansion using Slater determinant states. We then derive the expression for the $\rho^{(M)}$ and its relation to the original expansion coefficients $\Gamma_{i_1 \cdots i_N}$ following \cite{GigenaRossignoli2021}. We also discuss certain properties of $\rho^{(M)}$, bipartite-like representation of the state, and its Schmidt decomposition.

In Section~\ref{sec:maximally entangled fermionic states}, we aim to answer the question: for what $D,N$ values can we find a fermionic state $|\Psi\rangle$ giving maximum $M$-body entanglement entropy for different bipartitions. It will be related to the existence of the solution of a specific system of equations. This section establishes the relation between fermionic states and hypergraphs. Here, we also introduce a special class of hypergraphs, called $t$-designs, and state some important properties of these $t$-designs.  The language of hypergraphs and $t$-designs will be used in obtaining analytical and numerical results. 

Section~\ref{sec:random states} deals with random states. We first introduce the WL ensemble and go over some of its important properties. Then, we will define trace-fixed WL ensemble as a representative of random $M$-body DMs. We will look at different limits of the underlying parameters $D,N,c$ and derive the form of distributions of $\rho^{(M)}$ eigenvalues. Finally, the analytical expressions for mean entropy will be presented.  All analytical expressions will be compared with numerical results. Overall, this section aims to provide a comprehensive analysis of the statistical properties of the $M$-body density matrix and the associated entanglement entropy for random fermionic states in the limit of large single-particle dimension.

Several long proofs are presented in the Appendices.

\section{Characterization of density matrix and entanglement of fermionic states}\label{sec:characterization of DM and entropy}

This section reviews the concept of a fermionic $M$-body reduced DM and its corresponding entanglement entropy. We also introduce the notion of a maximally entangled fermionic state and discuss some important properties of these states.

\subsection{The \texorpdfstring{$M$}{M}-body DM and entanglement entropy review}\label{subsec:M-body DM}

For an $N$-fermion system, a general pure state $|\Psi\rangle$ can be expanded in the basis of Slater Determinants (SDs) $|i_1 i_2 \cdots i_N\rangle \equiv c^{\dagger}_{i_1}c^{\dagger}_{i_2} \cdots c^{\dagger}_{i_N}|0\rangle$ and written as ~\cite{GigenaRossignoli2021}:
\begin{equation} \label{eq:state}
|\Psi\rangle=\frac{1}{N!}\displaystyle{\sum\limits_{i_1,\dots,i_N=1}^{D}\Gamma_{i_1\cdots i_N} c^{\dagger}_{i_1}\cdots c^{\dagger}_{i_N}|0\rangle}.
\end{equation}
Here, $c^{\dagger}_i$ and $c_i$ act in a $D$-dimensional \textit{single-particle} Hilbert space to create and annihilate a single fermion in the single-particle state $i=1,\dots,D$, respectively. They satisfy standard fermionic anti-commutation relations, $\{c^{\dagger}_{i}, c_{j}\} = \delta_{ij}$. The dimensionality of the $N$-fermion Hilbert space $\mathcal{H}^{(D,N)}$ is $\binom{D}{N} = \frac{D!}{N! \, (D-N)!}$.  The complex coefficients $\Gamma_{i_1 \cdots i_N}$ are the elements of a fully antisymmetric rank-$N$ tensor and the normalization of the state corresponds to:
\begin{equation} \label{eq:norm}
\begin{aligned}
\frac{1}{N!}\displaystyle{\sum\limits_{i_1,\dots,i_N=1}^{D}|\Gamma_{i_1 \cdots i_N}|^2} = 1.
\end{aligned}
\end{equation}

Our primary object of interest is the $M$-body density matrix (DM) $\rho^{(M)}$ of an $N$-fermion state $|\Psi\rangle$ ($M=1,2 \dots, N$). The entries of $\rho^{(M)}$ (introduced in Eq.~(\ref{eq:definition of M-body DM}) and repeated here for convenience) are given by $\rho_{\alpha \alpha'}^{(M)} = \langle\Psi|C^{(M)\dagger}_{\alpha'} C^{(M)}_{\alpha} |\Psi\rangle$, where $C^{(M)\dagger}_{\alpha} = c^{\dagger}_{i_1} \cdots c^{\dagger}_{i_M}$ is an $M$-fermion creation operator. The collective index $\alpha$ labels all $\binom{D}{M}$ ordered combinations of the form $[i_1, \dots ,i_M]$ such that $i_1 < i_2 < \cdots < i_M$. To obtain a bipartite-like representation of $|\Psi\rangle$, we can rewrite Eq.~(\ref{eq:state}) as~\cite{GigenaRossignoli2021}:
\begin{equation} \label{eq:bipartite representation}
\begin{aligned}
|\Psi\rangle = \binom{N}{M}^{-1}\sum_{\alpha,\beta}\Gamma^{(M)}_{\alpha\beta} \, C^{(M)\dagger}_{\alpha} \, C^{(N-M)\dagger}_{\beta}|0\rangle,
\end{aligned}
\end{equation}
where $\Gamma^{(M)}_{\alpha\beta}$ = $\Gamma_{i_1\cdots i_M j_1 \cdots j_{N-M}}$ and $\alpha=[i_1, \dots ,i_M], \, \beta=[j_1, \dots ,j_{N-M}]$. We can view $\Gamma^{(M)}_{\alpha\beta}$ as the entries of an $n \equiv \binom{D}{M}$ by $m \equiv \binom{D}{N-M}$ complex matrix $\Gamma^{(M)}$. One can show~\cite{GigenaRossignoli2021} that $\rho^{(M)}$ is related to $\Gamma^{(M)}$ by:
\begin{equation} \label{eq:gamma gamma dagger}
\begin{aligned}
\rho^{(M)} = \Gamma^{(M)} \Gamma^{(M)\dagger}.
\end{aligned}
\end{equation}
Thus, $\rho^{(M)}$ defines the reduced DM of $M$ fermions when we ``trace out" the remaining $N-M$ fermions. It is Hermitian and positive semi-definite with $\text{Tr}\rho^{(M)} = \binom{N}{M}$. We can now define the \textit{$M$-body entanglement entropy} by assigning the von Neumann entropy to the DM $\rho^{(M)}$
\begin{equation}\label{eq:entropy}
\begin{aligned}
S( \rho^{(M)})=-\text{Tr}(\rho^{(M)} \ln{\rho^{(M)}}).
\end{aligned}
\end{equation}
Importantly, both the non-zero eigenvalues and the trace of $\rho^{(M)}$ are the same as those of $\rho^{(N-M)}$. To see this, one can perform the singular value decomposition (SVD) of $\Gamma^{(M)}$
\begin{equation} \label{eq:SVD on gamma^M}
\begin{aligned}
\Gamma^{(M)} = U^{(M)}D^{(M)}V^{(N-M) \dagger},
\end{aligned}
\end{equation}
where $U^{(M)},V^{(N-M)}$ are $\binom{D}{M} \times \binom{D}{M}$ and $\binom{D}{N-M} \times \binom{D}{N-M}$ unitary matrices and $D^{(M)}_{\nu \nu'} = \sqrt{\lambda^{(M)}_{\nu}} \delta_{\nu \nu'}$ is a $\binom{D}{M} \times \binom{D}{N-M}$ diagonal matrix. $\lambda^{(M)}_{\nu}$ denote the eigenvalues of $\rho^{(M)}$ or equivalently $\rho^{(N-M)}$. Using Eq.~(\ref{eq:SVD on gamma^M}), one can rewrite Eq.~(\ref{eq:bipartite representation}) in Schmidt-like diagonal form
\begin{equation} \label{eq:Schmidt decomposition}
\begin{aligned}
|\Psi\rangle = \binom{N}{M}^{-1}\sum\limits_{\nu=1}^{n_R} \sqrt{\lambda^{(M)}_{\nu}} \, A^{(M)\dagger}_{\nu} \, B^{(N-M)\dagger}_{\nu}|0\rangle,
\end{aligned}
\end{equation}
where $n_R$ is the rank of $\Gamma^{(M)}$ while $A^{(M)\dagger}_{\nu} = \sum_{\alpha} U^{(M)}_{\alpha \nu} C^{(M)\dagger}_{(\alpha)}$ and $B^{(N-M)\dagger}_{\nu} = \sum_{\beta} V^{(N-M)*}_{\beta \nu} C^{(N-M)\dagger}_{(\beta)}$ are ``normal'' operators, unitarily related to original creation operators, which diagonalize $\rho^{(M)}$ and $\rho^{(N-M)}$ respectively~\cite{GigenaRossignoli2021}. Since both DMs have the same spectrum, except for the number of zero eigenvalues, Eq.~(\ref{eq:entropy}) gives $S(\rho^{(M)}) = S(\rho^{(N-M)})$. This also provides a visualization, as shown in Fig.~\ref{fig:summary Fock space}, of $\rho^{(M)}$ as a mixture of $n_R$ pure $M$-fermion states, $A^{(M)\dagger}_{\nu} \ket{0}$, in the $M$-fermion sector of the Fock space. Similarly, $\rho^{(N-M)}$ corresponds to a mixed state of $n_R$ pure $(N-M)$-fermion states, $B^{(N-M)\dagger}_{\nu} \ket{0}$, living in the $(N-M)$-fermion sector of the Fock space.
Hence, each part of the bipartition carries the same information about the $M$-body entanglement of the state. Therefore, we will only consider $M\leq N/2$, unless otherwise stated.  

By choosing the convention $\text{Tr}\rho^{(M)} = \binom{N}{M}$, the SD states of the form $|\text{SD}\rangle = c^{\dagger}_{i_1} \cdots c^{\dagger}_{i_N} |0\rangle$ correspond to a $\rho^{(M)}$ with all $\binom{N}{M}$ non-zero eigenvalues equal to 1. It thus follows from Eq.~(\ref{eq:entropy}) that $S(\rho^{(M)}_\text{SD}) = 0$ for any $M$. In a fermionic system, SD states can hence be thought of as the simplest unentangled states. The other common approach is to look at the normalized density matrix $\rho^{(M)}_{n}=\rho^{(M)}/{\binom{N}{M}}$ for which $\text{Tr}\rho^{(M)}_n=1$. The von Neumann entropy for such a normalized DM is related to that of the un-normalized DM by:
\begin{equation}\label{eq:normalized DM entropy}
\begin{aligned}
S(\rho^{(M)}_n) =S( \rho^{(M)})/\binom{N}{M} +\ln{\binom{N}{M}}.
\end{aligned}
\end{equation}
Then, for SD states $|\text{SD}\rangle$, $S(\rho^{(M)}_{n, \text{SD}}) = \ln{\binom{N}{M}}$ becomes the lower bound for $M$-body entanglement entropy of fermionic states as described in \cite{Lemm2017, Coleman1963}.

\subsection{Maximum \texorpdfstring{$M$}{M}-body entanglement}\label{subsec:Maximum M-body entanglement}

As an example of the above formalism, consider the case of filling fraction $\frac{1}{r}$, i.e. $D=rN$ with $r\in \mathbb{N}$. We can then expand a generalized GHZ state using $r$ SDs as:
\begin{equation}\label{eq:GHZstate}
\begin{aligned}
|\Psi_{\text{GHZ}_r}\rangle&=\frac{1}{\sqrt{r}} \, \left\{\left|1, 2, \dots ,\frac{D}{r}\right\rangle + \left|\frac{D}{r}+1, \dots ,2\frac{D}{r}\right\rangle \right.\\
& \left. + \left|2\frac{D}{r}+1, \dots ,3\frac{D}{r}\right\rangle + \cdots \right.\\ 
&\left. + \left|(r-1)\frac{D}{r}+1, \dots ,D\right\rangle \right\}. 
\end{aligned}
\end{equation}
When $r=2$, we recover the standard GHZ state~\cite{GreenbergerZeilinger2007}.
For concreteness, we consider the 1-body DM  with entries $\rho_{ij}^{(1)} = \langle\Psi|c^{\dagger}_j c_i |\Psi\rangle$. Its eigenvalues lie in the interval $[0,1]$, and $\text{Tr}(\rho^{(1)})=N$. We can clearly see that
$(\rho^{(1)}_{\text{GHZ}_r})_{ij} =\langle\Psi_{\text{GHZ}_r}|c^{\dagger}_j c_i |\Psi_{\text{GHZ}_r}\rangle = \frac{1}{r}\delta_{ij}$. Hence, $\rho^{(1)}_{\text{GHZ}_r}$ is proportional to the identity with all $D$ eigenvalues being non-zero and equal to $\frac{1}{r}$: $\rho^{(1)}_{\text{GHZ}_r}=\frac{1}{r} \mathbb{I}$.   

As another example, we can consider a different state in the case of even single-particle dimension $D$. Following~\cite{GigenaRossignoli2021}, we use the collective pair creation operator
\begin{equation}\label{eq:collectiveopp}
\begin{aligned}
A^{\dagger} = \dfrac{1}{\sqrt{D/2}} \sum\limits_{i=1}^{D/2} c^{\dagger}_{2i-1} c^{\dagger}_{2i},
\end{aligned}
\end{equation}
to define the state $|\Psi_{2k}\rangle$ for $0 \leq k \leq D/2$
\begin{equation}\label{eq:Psi_2kstate}
\begin{aligned}
|\Psi_{2k}\rangle = \dfrac{(A^{\dagger})^k}{k!\sqrt{(2/D)^k \binom{D/2}{k}}} |0\rangle.
\end{aligned}
\end{equation}
Since the state contains $k$ pair creation operators $A^{\dagger}$, the number of fermions is $N=2k$. One can then show that the 1-body DM for $|\Psi_{2k}\rangle$ is~\cite{GigenaRossignoli2021}
\begin{equation} \label{eq:Psi_2kDM}
\begin{aligned}
\rho^{(1)}_{2k} = \frac{2k}{D} \mathbb{I}. 
\end{aligned}
\end{equation}
Thus, for $D=rN$, we get the same 1-body DM for both $|\Psi_{\text{GHZ}_r}\rangle$ and $|\Psi_{2k}\rangle$. Using Eq.~(\ref{eq:entropy}), the entropy for both states is $S(\rho^{(1)}) = N\ln{\frac{D}{N}}$ which, in the case of 1-body DM with the given $D,N$ values, corresponds to the maximum 1-body entanglement entropy [see Eq.~(\ref{eq:S maximum entropy general}) below]. 

This example shows the inability of 1-body DM $\rho^{(1)}$ to differentiate between two states. Thus, as outlined in~\cite{GigenaRossignoli2021}, one needs to consider the $M$-body DMs (for $2 \leq M \leq N/2$) to capture the differences in the entanglement structure between two states. As was shown  in~\cite{GigenaRossignoli2021}, when $M=2$, the eigenvalues of $\rho^{(2)}$ for $|\Psi_{\text{GHZ}_r}\rangle$ and $|\Psi_{2k}\rangle$ are indeed different, leading to different $2$-body entanglement.

For a general value of $M$, if the $M$-body DM of a state $|\Psi\rangle$ has the form $\rho_{\Psi}^{(M)} = \frac{\text{Tr}\rho_{\Psi}^{(M)}}{n} \mathbb{I}_{n \times n}$, then it gives maximum $M$-body entanglement entropy, or simply, \textit{maximum entropy} in Eq.~(\ref{eq:entropy})
\begin{equation} \label{eq:S maximum entropy general}
\begin{aligned}
     S(\rho_{\Psi}^{(M)}) \equiv S^{(M)}_\text{max} = \text{Tr}\rho^{(M)} \ln{\left( \frac{n}{\text{Tr}\rho^{(M)}}\right)},
\end{aligned}
\end{equation}
and the state can be called \textit{maximally $M$-body entangled}. In the case of $M=1$, we recover the maximum entropy for 1-body DM, $S^{(1)}_\text{max} = N \ln{\frac{D}{N}}$. To relate it to the standard result for maximum entropy~\cite{ZozulyaSchoutens2008}, we need to consider the normalized DM whose trace is 1, for which Eq.~(\ref{eq:S maximum entropy general}) becomes $S^{(M)}_\text{max} = \ln{n}$, where $n = \binom{D}{M}$ is a dimension of the $M$-fermion sector of Fock space.

The maximum entropy states satisfy a nesting property of entanglement, introduced in the following theorem: 

\phantomsection
\hypertarget{nesting-theorem}{\textbf{\textsc{Nesting property of a maximally entangled state:}}\label{nesting-theorem} 
    \textit{If a state $|\Psi\rangle$ is maximally $M$-body entangled, then it also is maximally $M'$-body entangled for $1 \leq M' \leq M$} (see Appendix \ref{sec:appendix proof of theorem M leads to M'} for the proof). 
This theorem also implies the contrapositive statement: if a state is not maximally $M'$-body entangled, then it is also not maximally $M$-body entangled. }

Finally, we demonstrate that the entanglement entropy for maximally entangled states satisfy the \textit{subadditivity} and \textit{strong subadditivity} inequalities. For the normalized DMs $\rho_n^{(M)}$, we can define the subadditivity of von Neumann entropy as:
\begin{equation} \label{eq:subadditivity}
\begin{aligned}
    S(\rho_{n}^{(M_1 + M_2)}) \leq S(\rho_{n}^{(M_1)}) + S(\rho_{n}^{(M_2)}),
\end{aligned}
\end{equation}
and the strong subadditivity as:
\begin{equation} \label{eq:strong subadditivity}
\begin{aligned}
    S(\rho_{n}^{(M_1 + M_2 + M_3)}) \leq & \,\, S(\rho_{n}^{(M_1 + M_3)}) + S(\rho_{n}^{(M_2 + M_3)})\\
    & - S(\rho_{n}^{(M_3)}),
\end{aligned}
\end{equation}
where each term in both inequalities represents the entropy of the corresponding normalized DMs, as defined in Eq.~(\ref{eq:entropy}). These are important relations that physically express the non-negativity of \textit{mutual information} and \textit{conditional mutual information}, respectively, and establish von Neumann entropy as a meaningful measure of many-body correlations. Therefore, the proof of these inequalities for quantum systems~\cite{LiebRuskai1973} represents a fundamental result in quantum information theory.  

In Appendix \ref{sec:appendix subadditivity and strong subadditivity}, we explicitly show the above two properties for maximally $M$-body entangled states. Strong subadditivity of other R\'enyi entropies was also shown to hold for fermionic Gaussian states in Ref.~\cite{CamiloEliens2019}.

\begin{figure*}[t]
    \centering
    \subfigure[]{    
        \centering
        \includegraphics[width=1.2\columnwidth]{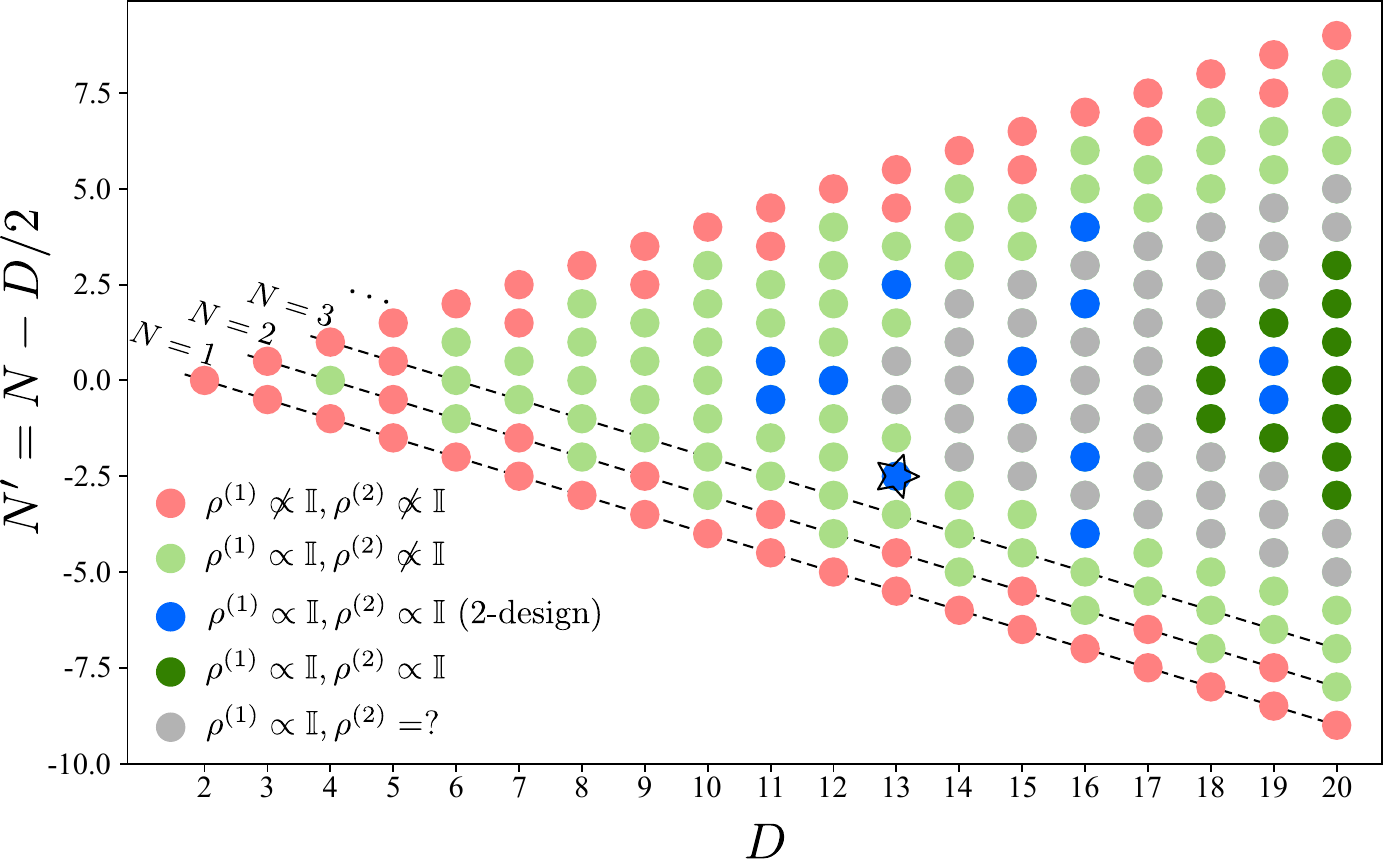}
        \label{fig:D and N mesh plot for M=1 and M=2}
    }
    \hspace{0 \columnwidth}
    \subfigure[]{
        \includegraphics[width=0.79\columnwidth]{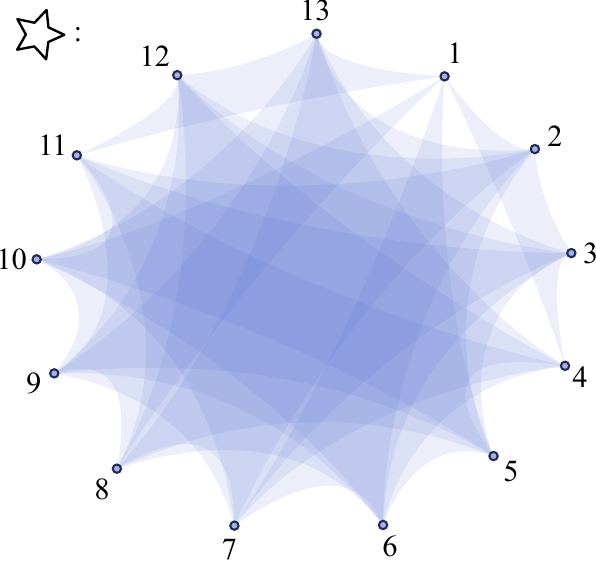}
        \label{fig:2-design as hypergraph}
    }
    \caption{Maximally many-body entangled fermionic states.
        (a) The existence of maximally $M$-body entangled (i.e., $\rho^{(M)}\propto\mathbb{I}$) $N$-fermion states in a $D$-dimensional single-particle Hilbert space, depending on the value of $D$ and $N$, for $M=1$ and $M=2$. The red dots 
        correspond to Hilbert spaces $\mathcal H^{(D,N)}$ which do not host maximally entangled states for either $\rho^{(1)}$ or $\rho^{(2)}$. The light green dots have maximally entangled states for $\rho^{(1)}$ but not for $\rho^{(2)}$. The blue and dark green dots have maximally entangled states for both $\rho^{(1)}$ and $\rho^{(2)}$. In particular, the maximally entangled states at the blue dots are associated with 2-designs (see Sec.~\ref{subsec:existence of maximally entangled states for given D,N} for a discussion on general $M$-designs giving maximally $M$-body entangled states). The gray dots have maximally entangled states for $\rho^{(1)}$ but are unknown for $\rho^{(2)}$, beyond our computational capabilities. The plot is symmetric about half-filling of the orbitals, $N=D/2$, due to particle-hole symmetry (see Appendix~\ref{sec:appendix particle-hole symmetry}). (b) The hypergraph associated with a maximally 2-body entangled $D=13, N=4$ state. The existence of such states is highlighted in (a) by a star. The hypergraph is a $2$-$(D=13,N=4,\lambda=1)$ design, hence a Steiner system (a special case of a hypergraph $t$-design, see Sec.~\ref{subsec:mini-review on hypergraphs}). Maximally $(N/2)$-body entangled states are all and only associated with Steiner systems (Sec.~\ref{subsec:existence of maximally entangled states for given D,N}). The hypergraph shown here is the only $2$-design Steiner system with $13$ vertices and each edge containing $4$ vertices. Therefore, it is the only hypergraph that yields maximally $2$-body entangled states for $D=13$ and $N=4$. The hypergraph has $b=13$ edges, corresponding to 13 constituent SDs of the maximally entangled states. The components of the states in each SD are all equal to $1/\sqrt{13}$ up to a phase (see Appendix \ref{sec:appendix M-designs give maximally entangled states}). In the case when $M=N/2$, only Steiner systems give rise to maximally entangled states. 
    }
    \label{fig: N and D mesh plots and 2-design}
\end{figure*}

\section{Maximally entangled \texorpdfstring{$N$}{N}-fermion states}\label{sec:maximally entangled fermionic states}

Our goal in this section is to determine whether there is an $N$-fermion state $\ket{\Psi} \in \mathcal{H}^{(D, N)}$ such that its reduced $M$-body density matrix is proportional to the identity matrix, $\rho_\Psi^{(M)} = \frac{\text{Tr}\rho^{(M)}_{\Psi}}{n}\mathbb{I}_{n\times n}$ where $n=\binom{D}{M}$ and $\text{Tr}\rho^{(M)}_{\Psi} = \binom{N}{M}$. This corresponds to $\ket{\Psi}$ being maximally $M$-body entangled.

We will relate this question to the problem of solving a particular system of linear equations in the hypergraph representation of the state $\ket{\Psi}$.  Specifically, for a given value of $M$, the coefficient matrix of the system of equations is exactly the \textit{$M^\text{th}$ incidence matrix} $A^{(M)}$ of the hypergraph associated with $\ket{\Psi}$.  Using this approach, we will show that \textit{$t$-designs} are maximally $M$-body entangled $N$-fermion states when $t = M$. Section \ref{subsec:mini-review on hypergraphs} includes a brief review of the theory of hypergraphs, $t$-designs and related topics.

\subsection{Criteria for maximally entangled states}\label{subsec:setting up the system of equations}

To begin, we note that if the DM is proportional to the identity it will remain so under any unitary transformation.  Hence, we can search for states that satisfy the above condition in the original basis of SDs neglecting equivalent rotations. 

For a general $N-$fermion state $\ket{\Psi}=\sum_{k=1}^b \alpha_k\ket{\text{SD}_k}$ ($\forall k, \alpha
_k\neq 0$), we will show that $\rho_\Psi^{(M)}$ being simply diagonal imposes a condition on its constituent SDs $\{\ket{\text{SD}_k}\}_{k=1}^b$.
Two SDs are said to have $p$ \textit{overlapping} orbitals if they have exactly $p$ creation operators in common.  Using the original definition of the $M$-body DM $\rho_{\alpha \beta}^{(M)} = \langle\Psi|C^{(M)\dagger}_{\beta} C^{(M)}_{\alpha} |\Psi\rangle$, we can see that any two SDs that have $N-M$ or more overlapping single-particle orbitals will produce non-zero off-diagonal elements in the $M$-body DM.  For example, when $N=3$ and $M=1$, SD states $|1 2 3\rangle$ and $|124\rangle$ give non-zero off-diagonal elements $\rho^{(1)}_{43}$ and $\rho^{(1)}_{34}$. Thus, in order to obtain a diagonal $M$-body DM, the following criterion needs to be satisfied:

\phantomsection
\hypertarget{overlap-criterion}{\textbf{\textsc{Overlap Criterion:}}\label{overlap-criterion} \textit{The number of overlapping orbitals for any pair of SDs in the expansion of a state $\ket{\Psi}$ with a diagonal $\rho_\Psi^{(M)}$ must be strictly less than $N-M$.}}

With this restriction in mind, let us first consider the case of $M = 1$ (1-body DM). Suppose we have a state $\ket{\Psi}$ which can be expanded into a linear combination of $b$ SDs satisfying the above criterion, $\ket{\Psi}=\sum_{k=1}^b \alpha_k\ket{\text{SD}_k}$. The $1$-body DM associated with $\ket{\Psi}$, will then be diagonal with entries given by
\begin{equation}\label{eq:Diagonal Elements of 1RDM in terms of Coefficients w.r.t. SDs}
    \rho_{ii}^{(1)} = \sum_{\substack{k, \\ c_i \ket{\text{SD}_k} \neq 0}} |\alpha_k|^2,
\end{equation}
where $i=1,\dots,D$ and the sum is taken over $k$ for which $\ket{\text{SD}_k}$ contains the $i^\text{th}$ orbital occupied. For instance, if the $i^\text{th}$ orbital is present in 3 terms in the SD expansion of $\ket{\Psi}$ which have coefficients $\alpha_{k_1}$, $\alpha_{k_2}$, and $\alpha_{k_3}$ respectively, then the $i^\text{th}$ diagonal element of $\rho^{(1)}_{\Psi}$ will be $|\alpha_{k_1}|^2 + |\alpha_{k_2}|^2 + |\alpha_{k_3}|^2$.

Now, if we want $\rho^{(1)}_{\Psi}$ to further be proportional to the identity with a proportionality constant $N/D$, we would require
\begin{equation} \label{eq:System of equations elementwise M=1}
\begin{aligned}
\rho^{(1)}_{ii} = \frac{N}{D}, \quad i = 1, \dots, D.
\end{aligned}
\end{equation}
Combining Eq.~(\ref{eq:Diagonal Elements of 1RDM in terms of Coefficients w.r.t. SDs}) and Eq.~(\ref{eq:System of equations elementwise M=1}), the condition on the coefficients of state $\ket{\Psi}$ and its constituent SDs can be expressed as system of equations
\begin{equation} \label{eq:System of equations M=1}
\begin{aligned}
A^{(1)} \Vec{x} = \frac{N}{D} \Vec{1}_D.
\end{aligned}
\end{equation}
Here $A^{(1)}$ is a $D\times b$ binary matrix with entries $A_{ik}^{(1)}=1$ if $c_i \ket{\text{SD}_k} \neq 0$ and 0 otherwise. $\Vec{x} = (|\alpha_1|^2, \dots, |\alpha_b|^2)^\text{T}$ is the $b$-dimensional vector of the norm-squared coefficients, and $\Vec{1}_D$ is the $D$-dimensional all-ones vector.

For a given set of SDs $\{\ket{\text{SD}_k}\}_{k=1}^b$ satisfying the \hyperref[overlap-criterion]{overlap criterion}, the matrix $A^{(1)}$ is determined, and Eq.~(\ref{eq:System of equations M=1}) can be viewed as a system of $D$ equations and $b$ unknowns. If one can find coefficients $\alpha_1,\dots,\alpha_b$ solving Eq.~(\ref{eq:System of equations M=1}), then the state $\ket{\Psi}=\sum_{k=1}^b \alpha_k\ket{\text{SD}_k}$ has $\rho^{(1)}_{\Psi} = \frac{N}{D} \mathbb{I}$. The problem of the existence of an $N$-fermion state $\ket{\Psi} \in \mathcal{H}^{(D, N)}$ with $\rho_\Psi^{(1)}=\frac{N}{D}\mathbb{I}$ can thus be answered by checking the existence of solutions to Eq.~(\ref{eq:System of equations M=1}) for all sets of SDs satisfying the overlap criterion.

We make two comments on Eq.~(\ref{eq:System of equations M=1}). First, any solution $\vec{x}$ to Eq.~(\ref{eq:System of equations M=1}) automatically satisfies the normalization condition $\sum_{k=1}^b x_k = \sum_{k=1}^b|\alpha_k|^2=1$ for the state $\ket{\Psi}$. This is a direct consequence of Eq.~(\ref{eq:System of equations elementwise M=1}), which amounts to choosing $\ket{\Psi}$ such that $\text{Tr}\rho_\Psi^{(1)} = \sum_{i=1}^{D} \rho^{(1)}_{ii} = N$, or equivalently, choosing a normalized $\ket{\Psi}$. Second, for our purpose, a solution $\vec{x}$ to Eq.~(\ref{eq:System of equations M=1}) must satisfy additional constraint: $x_k=|\alpha_k|^2\geq 0 \quad \forall k\in\{1,\dots, b\}$. Combined with the condition $\sum_k x_k=1$, the existence of such a solution is equivalent to determining whether $\frac{N}{D}\Vec{1}_D$ lies within the \textit{convex hull} of the column vectors of the coefficient matrix $A^{(1)}$. We use \textit{linear programming} methods to address this question. 

We now turn to the general $M$-body DM $\rho^{(M)}$ and extend the approach previously applied to the 1-body DM. For a state $\ket{\Psi}=\sum_{k=1}^b \alpha_k\ket{\text{SD}_k}$ with $\{\ket{\text{SD}_k}\}_{k=1}^b$ satisfying the \hyperref[overlap-criterion]{overlap criterion}, $\rho_\Psi^{(M)}$ is diagonal with diagonal entries
\begin{equation}\label{eq:Diagonal Elements of MRDM in terms of Coefficients w.r.t. SDs}
    \rho_{\beta\beta}^{(M)} = \sum_{\substack{k, \\ C_\beta^{(M)} \ket{\text{SD}_k} \neq 0}} |\alpha_k|^2,
\end{equation}
where $\beta$ is an $M$-element subset ($M$-subset) of the orbitals $\{1,\dots,D\}$ and the sum is taken over $k$ for which $\ket{\text{SD}_k}$ contains all the orbitals in $\beta$.
By further requiring $\rho_{\Psi}^{(M)} = \frac{\text{Tr}\rho_{\Psi}^{(M)}}{n} \mathbb{I}_{n \times n}$, where $n=\binom{D}{M}$ and $\text{Tr}\rho^{(M)}_{\Psi}=\binom{N}{M}$, we impose the following condition on $\ket{\Psi}$:
\begin{equation} \label{eq:System of equations general M}
\begin{aligned}
A^{(M)} \Vec{x} = \frac{\binom{N}{M}}{\binom{D}{M}} \Vec{1}_{\binom{D}{M}},
\end{aligned}
\end{equation}
where $A^{(M)}$ is a $\binom{D}{M} \times b$ binary matrix with entries $A_{\beta k}^{(M)}=1$ if $C_\beta^{(M)} \ket{\text{SD}_k} \neq 0$ and 0 otherwise, $\Vec{x} = (|\alpha_1|^2, \dots, |\alpha_b|^2)^\text{T}$, and $\vec{1}_{\binom{D}{M}}$ is the $\binom{D}{M}$-dimensional all-ones vector. Given a set of SDs satisfying the overlap criterion, the matrix $A^{(M)}$ is determined, and Eq.~(\ref{eq:System of equations general M}) can be viewed as a system of $\binom{D}{M}$ equations and $b$ unknowns, where the unknowns are subject to the constraints $\sum_k x_k=1$ and $x_k\geq 0 \quad \forall k\in\{1,\dots, b\}$. Again, if the system of equations has a solution, then a state $\ket{\Psi}$ with $\rho^{(M)}_{\Psi} \propto \mathbb{I}_{n \times n}$ can be constructed from the given set of SDs. Such a solution exists if and only if $\binom{N}{M}/\binom{D}{M}\vec{1}_{\binom{D}{M}}$ is in the convex hull of the column vectors of the coefficient matrix $A^{(M)}$.

\subsection{A brief review of hypergraphs and their relation to fermionic states}\label{subsec:mini-review on hypergraphs}

In general, for fixed $D$ and $N$, there are many sets of SDs satisfying the \hyperref[overlap-criterion]{overlap criterion}. Thus, to determine whether a state with $\rho^{(M)}_{\Psi} \propto \mathbb{I}_{n \times n}$ exists for $(D,N)$, we would need to, in principle, explore every such set of SDs and check whether a solution to Eq.~(\ref{eq:System of equations general M}) exists. This task becomes computationally challenging because the number of possible sets of SDs grows exponentially with the number of SDs in the set.

Fortunately, we can dramatically reduce the search space by identifying a hypergraph structure with each set of SDs. We will show that all sets of SDs whose corresponding hypergraphs fall in the same \textit{isomorphism} class give rise to equivalent systems of equations. Therefore, we only need to consider a representative element of each isomorphism class when searching for solutions. In what follows, we briefly review hypergraphs and establish their relation to sets of SDs.

An (undirected) \textit{hypergraph} is a set $V$, whose elements are called vertices, together with a set $E$, whose elements are called edges, where each edge $e\in E$ is a subset of $V$. If all the edges in a hypergraph contain exactly $N$ vertices, then the hypergraph is called \textit{$N$-uniform}. Associated with each set of $N$-fermion SDs $\{\ket{\text{SD}_k}\}_{k=1}^b$ of single-particle dimension $D$, there is a natural $N$-uniform hypergraph of $D$ vertices and $b$ edges, where the vertices correspond to single-particle orbitals and the edges correspond to the SDs.

A hypergraph of $D$ vertices $V=\{v_1,\dots,v_D\}$ and $b$ edges $E=\{e_1,\dots,e_b\}$ can be represented by its \textit{incidence matrix} $A$. It is a $D \times b$ binary matrix with entries $A_{ik}=1$ if $v_{i}\in e_k$ and 0 otherwise. As a generalization of the incidence matrix, the \textit{$M^\text{th}$ incidence matrix} of a hypergraph is a $\binom{D}{M} \times b$ binary matrix $A^{(M)}$ with entries $A^{(M)}_{\beta k}$=1 if $\beta\subseteq e_k$ and 0 otherwise, where $\beta$ is an $M$-element subset of $V$. The $1^\text{st}$ incidence matrix $A^{(1)}$ is simply the incidence matrix $A$.  The $M^\text{th}$ incidence matrix $A^{(M)}$ of the hypergraph associated to a set of SDs $\{\ket{\text{SD}_k}\}_{k=1}^b$ thus has entries $A_{\beta k}^{(M)}=1$ if $C_\beta^{(M)} \ket{\text{SD}_k} \neq 0$, and 0 otherwise, where $\beta$ is an $M$-subset of the orbitals $\{1,\dots,D\}$. For a given set of SDs, we then recognize the coefficient matrix in Eq.~(\ref{eq:System of equations general M}) as the $M^\text{th}$ incidence matrix of the hypergraph associated with the set of SDs.

An $N$-uniform hypergraph of $D$ vertices is known as a \textit{$t$-$(D,N,\lambda)$ design}, a \textit{$t$-design} for short, if each $t$-subset of the vertex set is contained in exactly $\lambda$ edges. A $t$-design is called a \textit{Steiner system} if each $t$-subset is contained in exactly one edge, i.e., if $\lambda=1$. 
In the language of sets of SDs, this means that every $t$-subset of orbitals appears in exactly one SD in the set. 

Two hypergraphs $(V_1, E_1)$ and $(V_2, E_2)$ are \textit{isomorphic} if there exists a bijection $f : V_1 \rightarrow V_2$ such that $\forall e \in E_1, \,\, f(e) \in E_2$, and the resulting map $f_*: E_1 \to E_2$ defined by $e \mapsto f(e)$ is also a bijection.  It follows that the $M^\text{th}$ incidence matrices of two isomorphic hypergraphs are related by row and column permutations. For two sets of SDs whose associated hypergraphs are isomorphic, their coefficient matrices in Eq.~(\ref{eq:System of equations general M}) are then also related by row and column permutations. Since all entries of the RHS of Eq.~(\ref{eq:System of equations general M}) are equal, a row permutation of the coefficient matrix does not change the solution space. Meanwhile, a column permutation of the coefficient matrix only changes the solution space by a permutation of the entries of $\vec{x}$, which does not change the existence of solutions to the equation.  Thus, we need only take one hypergraph from each isomorphism class and solve Eq.~(\ref{eq:System of equations general M}) for the corresponding set of SDs to exhaustively confirm the existence of a solution.

We make two comments on the hypergraphs associated with sets of SDs. First, we can represent not only a set of SDs but also a generic $N$-fermion state with an $N$-uniform hypergraph by assigning to each edge a complex weight, see Fig. \ref{fig:summary hypergraph}. Second, in terms of the associated hypergraphs, the \hyperref[overlap-criterion]{overlap criterion} for a set of SDs is equivalent to any two edges in the hypergraph sharing fewer than $N-M$ vertices.

\subsection{Searching for maximally entangled states using hypergraph isomorphism}\label{subsec:hypergraph isomorphism numercis method}

We now present the general procedure for searching for states with $\rho^{(M)}_{\Psi} \propto \mathbb{I}$. For given $D, N, M$, we enumerate isomorphism classes of $N$-uniform hypergraphs with $D$ vertices satisfying the \hyperref[overlap-criterion]{overlap criterion}. We start by generating hypergraphs with at least $b_\text{min}= \binom{D}{M}/ \binom{N}{M}$ edges. To see why this gives a lower bound for the number of edges in our search, consider Eq.~(\ref{eq:System of equations general M}), which has a solution only when each row of $A^{(M)}$ is non-zero. This corresponds to each $M$-subset of vertices appearing in at least one edge. There are in total $\binom{D}{M}$ $M$-subsets, and each edge has $\binom{N}{M}$ $M$-subsets, so we need at least $\binom{D}{M}/\binom{N}{M}$ edges for each $M$-subset to be in at least one edge. For each isomorphism class generated, we check if Eq.~(\ref{eq:System of equations general M}) has a solution using linear programming methods. We repeat the process for hypergraph isomorphism classes with more edges until we either cannot add another edge without breaking the overlap criterion or we find a solution. If a solution is found, then the given $D, N$ has a state with $\rho^{(M)}_{\Psi} \propto \mathbb{I}$; otherwise, such a state does not exist.

Hypergraph isomorphism classes were enumerated using the mathematics software system SageMath \cite{sagemath} and the software package \textsc{nauty}~\cite{McKayPiperno2014}. Despite the considerable boost in performance gained by working with isomorphism classes, the total number of isomorphism classes that need to be checked is still sensitive to $(D,N)$ and grows rapidly as $D,N$ increases.  As such, this approach is limited to fairly small $D,N$ values.

\subsection{Existence of maximally entangled states for given \texorpdfstring{$D,N$}{D, N}}\label{subsec:existence of maximally entangled states for given D,N}

Given a $D$-dimensional single-particle Hilbert space, we find that the existence of an $N$-fermion state with $\rho^{(M)}\propto\mathbb{I}$ depends on the values of $D$, $N$, and $M$. The numerical results for $\rho^{(1)}$ and $\rho^{(2)}$, obtained by following the algorithm of section \ref{subsec:hypergraph isomorphism numercis method}, are shown in Fig.~\ref{fig:D and N mesh plot for M=1 and M=2}.

We first point out two general features regarding the existence of maximally entangled states. First, for given $D$, $N$, and $M$, if an $M$-$(D,N,\lambda)$ design (see Sec.~\ref{subsec:mini-review on hypergraphs}) satisfying the the \hyperref[overlap-criterion]{overlap criterion} exists, then the $D,N$ values host states with $\rho^{(M)}\propto\mathbb{I}$ [see the blue dots in Fig.~\ref{fig:D and N mesh plot for M=1 and M=2} and a particular 2-design in Fig.~\ref{fig:2-design as hypergraph}]. In Appendix \ref{sec:appendix M-designs give maximally entangled states}, we show that these states can be constructed from the set of SDs corresponding to the $M$-design, with coefficients all equal to $1/\sqrt{b}$ up to a phase, where $b$ is the number of SDs in the set. Second, the existence of maximally entangled states exhibits particle-hole symmetry. Specifically, the presence (or absence) of an $N$-fermion state $\ket{\Psi}$ with $\rho_\Psi^{(M)}\propto \mathbb{I}$ implies the presence (or absence) of a $(D-N)$-fermion state $\ket{\ols{\Psi}}$ with $\rho_{\ols{\Psi}}^{(M)}\propto \mathbb{I}$. A proof of this result is provided in Appendix \ref{sec:appendix particle-hole symmetry}. The particle-hole symmetry is evident in the symmetry of the plot in Fig.~\ref{fig:D and N mesh plot for M=1 and M=2} about half-filling of the orbitals, $N=D/2$.

Next, we examine how the existence of maximally entangled states depends on the relationship between $N$ and $M$, 
considering three cases: (i) $N<2M$, (ii) $N=2M$, and (iii) $N>2M$. 
In the first case, when $N<2M$, we have 
$\text{dim}(\rho^{(N-M)})=\binom{D}{N-M}<\binom{D}{M}=\text{dim}(\rho^{(M)})$. From the Schmidt decomposition in Eq.~(\ref{eq:Schmidt decomposition}), we know that $\rho^{(M)}$ and $\rho^{(N-M)}$ have the same non-zero eigenvalues. Consequently, in this case, $\rho^{(M)}$ has zero eigenvalues, and it cannot be proportional to the identity. 
Therefore, there is no maximally $M$-body entangled state when $N < 2M$.  As an example, consider the $N = 1$ line in Fig.~\ref{fig:D and N mesh plot for M=1 and M=2} which hosts no solutions since $N < 2M$ for all $M \ge 1$. 
In the second case, when $N=2M\Leftrightarrow M=N/2$, a state has $\rho^{(M)}\propto\mathbb{I}$ only when any pair of its constituent SDs share fewer than $N-M=M$ overlapping orbitals. This implies that any $M$-subset of orbitals must appear in no more than one SD. On the other hand, Eq.~(\ref{eq:System of equations general M}) has a solution only when $A^{(M)}$ has no zero rows, meaning that each $M$-subset of orbitals must appear in at least one SD. Therefore, when $N=2M$, the only states with $\rho^{(M)}\propto\mathbb{I}$ are constructed from sets of SDs in which each $M$-subset of orbitals is contained in exactly one SD, i.e., whose associated hypergraphs are Steiner systems~\cite{Note2}.  For example, consider the blue dots with $N=4$ in Fig.~\ref{fig:D and N mesh plot for M=1 and M=2}. 

In the third case, when $N>2M\Leftrightarrow M<N/2$, we do not have a general statement as in the previous two cases, and a state with $\rho^{(M)}\propto \mathbb{I}$ might or might not exist. 

In summary, the $N$-fermion state space $\mathcal{H}^{(D, N)}$ of single-particle dimension $D$
\begin{enumerate}
    \item does not contain a state with $\rho^{(M)}\propto\mathbb{I}$ when $N<2M$ or, by particle-hole symmetry, when $N > D - 2M$, 
    \item contains a state with $\rho^{(M)}\propto\mathbb{I}$ if and only if the associated hypergraph of the state is a Steiner system, when $N=2M$ or $N = D-2M$, and 
    \item might contain a state with $\rho^{(M)}\propto\mathbb{I}$ when $2M < N < D-2M$. 
\end{enumerate}
These three statements are corroborated by our results in Fig.~\ref{fig:D and N mesh plot for M=1 and M=2} and sketched in Fig.~\ref{fig:summary D-N plane}. 
Thus, in the $D$-$N$ plane, maximally entangled states may only exist in the triangular region defined by $ 2M \leq N \leq D - 2M $.  
Thus, the region hosting $\rho^{(M)}\propto\mathbb{I}$ states shrinks as $M$ increases, cf. $\rho^{(1)}$ and $\rho^{(2)}$ in Fig.~\ref{fig:D and N mesh plot for M=1 and M=2}. 

Finally, we comment on the specific cases where $N=2$ or $N=D-2$, and $M=1$. In these specific cases, a maximally entangled state exists when $D$ is even, but no such state exists when $D$ is odd.
To understand this pattern, we recall the \hyperref[overlap-criterion]{overlap criterion} that for an $N$-fermion state to have $\rho^{(M)}\propto\mathbb{I}$, any pair of its constituent SDs must have less than $N-M$ overlapping orbitals. In the case when $N=2$ and $M=1$, any overlap is forbidden. When $D$ is even, this is always satisfied by a 2-particle GHZ state with filling fraction $2/D$ [see Eq.~(\ref{eq:GHZstate})], which we have shown to have $\rho^{(1)}\propto\mathbb{I}$ in Section \ref{subsec:Maximum M-body entanglement}. When $D$ is odd, since a set of non-overlapping 2-fermion SDs can only contain an even number of single-particle orbitals, there is at least one single-particle orbital that is not present in any SD in the set. This implies that the $A^{(1)}$ matrix of the set has a zero row and Eq.~(\ref{eq:System of equations M=1}) does not have a solution; the missing orbital will similarly have a vanishing diagonal element in $\rho^{(1)}$. Therefore, there is no 2-fermion state with $\rho^{(1)}\propto \mathbb{I}$ for odd $D$.  In contrast, for even $D$, the generalized GHZ states always have $\rho^{(1)}\propto \mathbb{I}$, causing the alternating pattern observed on the $N=2$ line [see Fig.~\ref{fig:D and N mesh plot for M=1 and M=2}]. Particle-hole symmetry results in the same pattern on the complementary $N = D - 2$ line.

\section{Random states}\label{sec:random states} 
With large single-particle dimensions $D \gtrsim 15$ [see Fig.~\ref{fig:D and N mesh plot for M=1 and M=2}], it is no longer feasible to look for maximally entangled states by searching through all hypergraph isomorphism classes as we outlined in Sec.~\ref{subsec:hypergraph isomorphism numercis method}. 
At the same time, even if a maximally $M$-body entangled state exists for given $D,N$, one may wonder if such states are a rare occurence (such as a $t$-design Steiner system, Sec.~\ref{subsec:existence of maximally entangled states for given D,N}) or if in fact most states are maximally entangled. 
To answer these questions, here we focus on the statistical behavior of the $M$-body density matrix $\rho^{(M)}$ and its associated entropy $S(\rho^{(M)})$ in the large-$D$ limit. 
    
Instead of analyzing individual states, we employ a large ensemble of randomly generated fermionic states to investigate their statistical properties.
We demonstrate that the statistics of $M$-body DMs are well-described by the trace-fixed Wishart-Laguerre (WL) random matrix ensemble. In the $D \gg 1$ limit, the eigenvalue distribution of $\rho^{(M)}$ approaches the rescaled Mar\v{c}enko-Pastur distribution with the parameter $c = \binom{D}{M}/\binom{D}{N-M}$ controlling its shape. 
We consider two limiting cases: $c=1$ (achieved when $M=N/2$, equal bipartitioning) and $c \to 0$ (achieved when $D \gg N > 2M$). We show that the eigenvalue distributions in these two cases are qualitatively different.  For $c = 1$, the distribution is broad while in the $c \to 0$ limit, the distribution becomes highly peaked around an eigenvalue corresponding to a maximum entropy state. This distinction also leads to different behaviors in the mean value of entropy for each case. We relate this difference to the fact that, for $M=N/2$, maximum entanglement entropy is achieved only by states related to Steiner systems and is therefore a rare occurrence, whereas when $M < N/2$ most states are maximally $M$-body entangled.

\subsection{Wishart-Laguerre (WL) Ensemble}\label{subsec:Classical WL ensemble}

The WL ensemble describes an ensemble of square $n \times n$ Wishart matrices $W=HH^{\dagger}$, where $H$ is an $n\times m$ matrix $(m \geq n)$ with independently and identically distributed entries from the Gaussian distribution $N(0,1)$.  Hence, the generated Hermitian matrix $W$ has correlated entries. Wishart matrices have $n$ non-negative eigenvalues since they are positive semidefinite. The joint probability density function (jpdf) of the entries of the Wishart matrix $W$ is      
\begin{equation} \label{eq:WL ensemble classic}
P_\text{WL}[W] = {\mathcal{N}}e^{-\frac{1}{2}\text{Tr}W + \gamma^2\text{Tr}(\ln{W})},
\end{equation}
where $\mathcal{N}$ is a normalization constant and $\gamma^2 = \frac{\beta}{2}(m-n+1)-1$. We take the Dyson index $\beta= 2$ to restrict $H$ to have complex entries. This choice fixes $\gamma^2 = m-n$ ~\cite{LivanVivo2018}.

The jpdf of the eigenvalues $x_i$ of $W$ can be derived from the jpdf of its entries, Eq.~(\ref{eq:WL ensemble classic}), by appending the \textit{Vandermonde} determinant~\cite{LivanVivo2018}:
\begin{equation} \label{eq:WL classic eigenval distribution}
P_\text{WL}(x_1, \dots , x_n) \propto e^{-\frac{1}{2}\sum\limits^{n}_{i=1} x_i + \gamma^2 \sum\limits^{n}_{i=1}\ln(x_i)} \prod\limits_{j<k} |x_j - x_k|^2.
\end{equation}
The average spectral density of $W$ is then given by marginalizing the jpdf Eq.~(\ref{eq:WL classic eigenval distribution}) over $n-1$ of the eigenvalues $x_i$, $P_{\text{WL}}(x) = \int_{0}^{\infty} \prod\limits_{j=1}^{n-1} dx_j P(x_1,\dots, x_n)$. The resulting form of $P_{\text{WL}}(x)$ is well known and can be found in Ref.~\cite{LivanVivo2011}. 

Here, we are interested in the asymptotic form of the spectral density $P_{\text{WL}}(x)$ in the large-$n$ (hence large-$m$ for $m \geq n$) limit. When $n\rightarrow \infty$ and when $c=n/m \leq 1$ is fixed, $P_{\text{WL}}(x)$ approaches $\frac{1}{2n}P_\text{MP}(\frac{x}{2n})$ ~\cite{LivanVivo2018}, where $P_\text{MP}$ is the Mar\v{c}enko-Pastur scaling function given by
\begin{equation} \label{eq:MP distribution}
P_\text{MP}(y) = \frac{1}{2\pi y}\sqrt{(y - \xi_{-})(\xi_{+} - y)},
\end{equation}
with ${\xi_{\pm} = (1 \pm 1/\sqrt{c})}^2$. Explicitly, with $c\leq 1$ being fixed, the spectral density of the WL ensemble has an asymptotic form in the $n\rightarrow\infty$ limit
\begin{equation} \label{eq:WL spectral density distribution}
P_{\text{WL}}(x) \overset{n \rightarrow \infty}{=} \frac{1}{2\pi x}\sqrt{\left(\frac{x}{2n} - \xi_{-}\right)\left(\xi_{+} - \frac{x}{2n} \right)}.
\end{equation}
The mean $\mu_{\text{WL}} = \langle x \rangle_{\text{WL}}$ and the standard deviation $\sigma_{\text{WL}}^2 = \langle x^2 \rangle_{\text{WL}} - \langle x \rangle_{\text{WL}}^2$ of the distribution in Eq.~(\ref{eq:WL spectral density distribution}) can be computed to be
\begin{equation} \label{eq:WL large n,m mean and std dev}
\mu_{\text{WL}} = \frac{2n}{c},~~\sigma_{\text{WL}} = \frac{2n}{\sqrt{c}}.
\end{equation}
In a further limit $c \rightarrow 0$, expanding $\xi_{\pm}\overset{c \rightarrow 0}{\approx} \frac{1}{c} \pm \frac{2}{\sqrt{c}}$ in Eq.~(\ref{eq:WL spectral density distribution}) reveals that $P_{\text{WL}}(x)$ approaches a semi-circular distribution
\begin{equation} \label{eq:WL SC distribution}
P_{\text{WL}}(x) \overset{\substack{n \to \infty \\ c \to 0}}{=} \frac{1}{2\pi \sigma_{\text{WL}}^2}\sqrt{4\sigma_{\text{WL}}^2 - (x-\mu_{\text{WL}})^2},
\end{equation}
with the same mean and standard deviation as in Eq.~(\ref{eq:WL large n,m mean and std dev}). After reviewing the WL ensemble, we will next examine its connection to the ensemble of density matrices and analyze the statistical properties of the latter.

\subsection{Trace-fixed WL ensemble in the \texorpdfstring{$c \to 0$}{c->0} limit and GSUE equivalence}\label{subsec:TFWL and GSUE equivalance}

We will demonstrate that the eigenvalue distribution of DMs $\rho^{(M)}$ matches the spectral density of a Wishart matrix ensemble with a fixed trace, $\text{Tr}\rho^{(M)}$. For that purpose, we first relate the two ensembles of matrices by relating the $\Gamma$ matrices comprising $\rho^{(M)}$ in Eq.~(\ref{eq:gamma gamma dagger}) to the $H$ matrices that compose the Wishart matrix $W$. For the $M$-body DMs $\rho^{(M)}$ of random $N$-fermion states with single-particle dimension $D$, the $\Gamma$ matrices have dimensions $\binom{D}{M}\times\binom{D}{N-M}$. Meanwhile, for a WL ensemble with parameters $n,m$ ($m\geq n$), the $H$ matrices have dimensions $n\times m$. This suggests that we consider a WL ensemble with parameters $n=\binom{D}{M}$ and $m=\binom{D}{N-M}$, and with a trace fixed to $\text{Tr}\rho^{(M)} = \binom{N}{M}$. The entries of the trace-fixed WL ensemble are distributed according to the jpdf
\begin{equation} \label{eq:WL ensemble trace-fixed}
P_\text{TWL}[W] = \mathcal{N}e^{-\frac{1}{2}\text{Tr}W + \gamma^2\text{Tr}(\ln{W})}\delta(\text{Tr}W-\text{Tr}\rho^{(M)}),
\end{equation}
where $\gamma^2 = m-n$.

To obtain the spectral density of the trace-fixed WL ensemble, we start from Eq.~(\ref{eq:WL ensemble trace-fixed}) and separate the trace part by writing $W=\frac{\text{Tr}W}{n} \mathbb{I}_{n \times n} + \delta W$, where $\text{Tr}(\delta W) = 0$. Note that since $W$ is Hermitian, $\delta W$ is also Hermitian, and both matrices are diagonalized by the same unitary transformation. Hence, due to the trace-fixing $\delta$-function, shifting the eigenvalues of $\delta W$ by $\text{Tr}\rho^{(M)}/n$, we can recover the eigenvalue distribution for $W$ itself. This allows us to look at the jpdf of eigenvalues of $\delta W$ instead of $W$. If we assume that $||\delta W||\ll\text{Tr}\rho^{(M)}/n$, where $||\delta W||$ is a spectral norm of $\delta W$, then by expanding $\ln{W}$ in the exponent of Eq.~(\ref{eq:WL ensemble trace-fixed}) up to order $\delta W^2$ and integrating out the trace, we get the jpdf of the entries of $\delta W$
\begin{equation} \label{eq:Assimptotic GSUE}
\begin{aligned}
P[\delta W] & = \int_{0}^{\infty} P_\text{TWL}[W] d (\text{Tr} W) \\
& \approx \tilde{\omega} \exp \left[-\frac{\gamma^2}{2}\text{Tr}\left([n/\text{Tr}\rho^{(M)}]^2 \delta W^2\right)\right],
\end{aligned}
\end{equation}
where 
\begin{equation} \label{eq:constant prefactor}
\begin{aligned}
\tilde\omega(\text{Tr}\rho^{(M)}) = {\mathcal{N}} \left(\frac{\text{Tr}\rho^{(M)}}{n}\right)^{\gamma^2 n} e^{-\frac{1}{2}\text{Tr}\rho^{(M)}},
\end{aligned}
\end{equation}
which is a constant after trace-fixing. The linear term, proportional to $\delta W$, vanishes since $\text{Tr}(\delta W) = 0$. The resulting jpdf of $\delta W$ is exactly that of the matrix elements of the Gaussian special unitary ensemble (GSUE). From the jpdf of the entries of $\delta W$, Eq.~(\ref{eq:Assimptotic GSUE}), we can easily obtain the jpdf of the eigenvalues $z_i$ of $\delta W$ by appending the Vandermonde determinant~\cite{LivanVivo2018} 
\begin{equation} \label{eq:Eigenvalue x jpdf}
\begin{aligned}
P(z_1,\dots,z_n) = \tilde{\omega} e^{-\frac{1}{2}\sum\limits_{i=1}^n (\gamma n z_i/\text{Tr}\rho^{(M)})^2} \prod\limits_{k<j}|z_k - z_j|^2.
\end{aligned}
\end{equation}

Again, we are interested in the limit $n=\binom{D}{M} \rightarrow \infty$. To obtain the spectral density of $\delta W$ in this limit, it is convenient to first apply a change of variables $y_i =  \frac{\gamma n}{\text{Tr}\rho^{(M)}} z_i$. The jpdf for the new variables $y_i$ is 
\begin{equation} \label{eq:Eigenvalue y jpdf}
\begin{aligned}
\tilde{P}(y_1,\dots,y_n) = \! \left[\frac{\text{Tr}\rho^{(M)}}{\gamma n}\right]^{\frac{n(n+1)}{2}} \! \tilde{\omega} e^{-\frac{1}{2}\sum\limits_{i=1}^n y_i^2} \prod\limits_{k<j}|y_k - y_j|^2.
\end{aligned}
\end{equation}
Taking the $n\rightarrow\infty$ limit, the $y_i$ variables take values from $(-\infty,\infty)$. Integrating out $n-1$ of the variables $y_i$ in Eq.~(\ref{eq:Eigenvalue y jpdf}), we find that the average spectral density of eigenvalues is in the form of a semi-circular distribution ~\cite{LivanVivo2018} 
\begin{equation} \label{eq:Eigenvalue y SC}
\begin{aligned}
\int_{-\infty}^\infty \prod_{j=1}^{n-1}dy_j\tilde{P}(y_1,\dots,y_n)   \overset{n\rightarrow\infty}{=}
\frac{1}{\pi \sqrt{2n}} \sqrt{2 - \frac{y^2}{2n}}.
\end{aligned}
\end{equation}
Finally, rescaling back to the variable $z$, we get the spectral density of $\delta W$ under the condition $||\delta W||\ll\text{Tr}\rho^{(M)}/n$ and in the $n\rightarrow \infty$ limit
\begin{equation} \label{eq:Eigenvalue x SC}
\begin{aligned}
P(z) = \frac{1}{2 \pi \sigma^2} \sqrt{4\sigma^2 - z^2},
\end{aligned}
\end{equation}

\begin{figure}[!t]
    \centering
    \subfigure[]{
        \includegraphics[width=1.0\columnwidth]{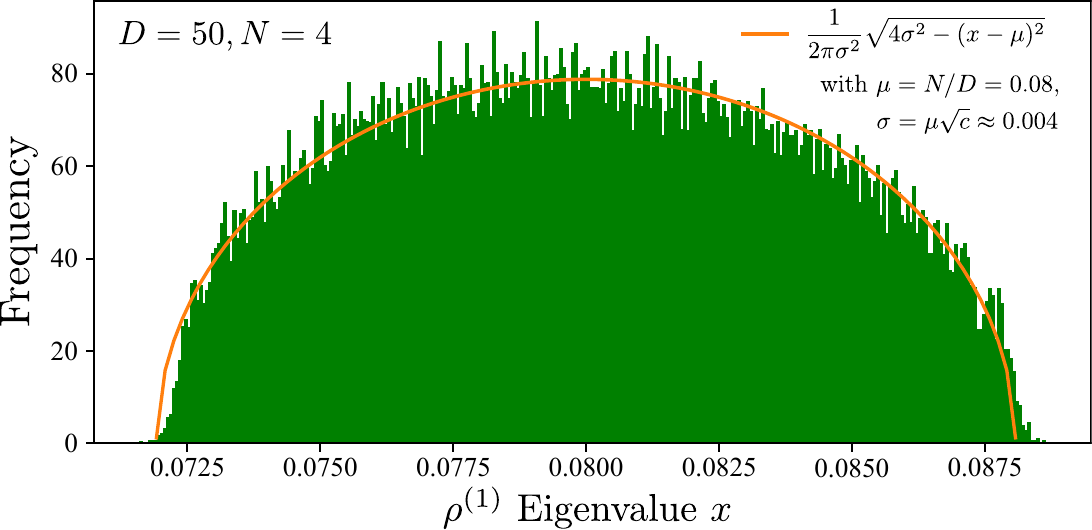}
        \label{fig:plot1 SC distribution D=50}
    }
    \subfigure[]{
        \includegraphics[width=1.0\columnwidth]{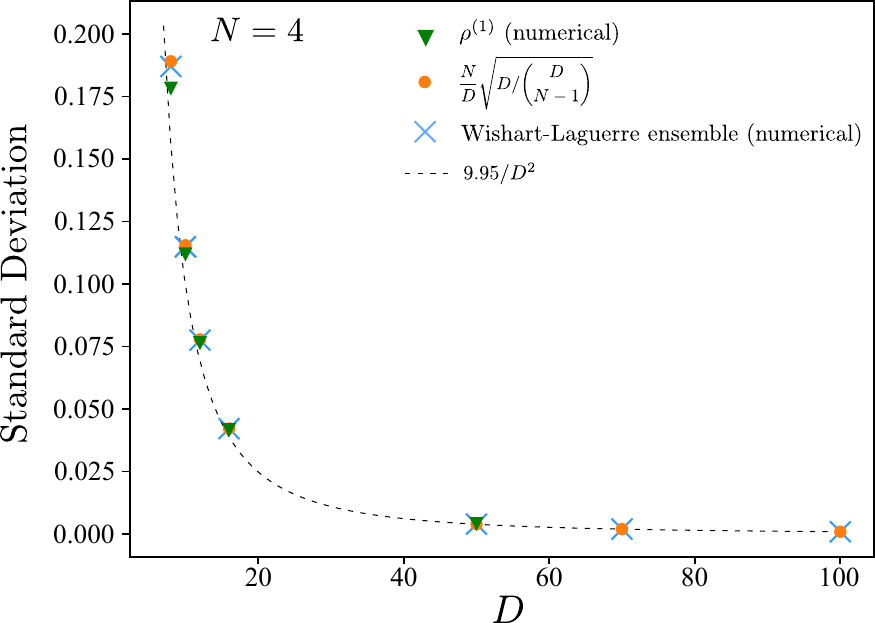}
        \label{fig:plot2 std deviation dependance on D}
    }
    \caption{One-body DM of random states. (a) The eigenvalue distribution of $\rho^{(1)}$ for randomly generated states with $D=50, N=4$. The numerical result obtained over 1000 realizations (green histogram) is well approximated by the semi-circular distribution [Eq.~(\ref{eq:Eigenvalue x SC final form})] with parameters $c=\binom{D}{M}/\binom{D}{N-M}\approx0.0025$, $\mu=N/D=0.08$, and $\sigma=\mu\sqrt{c}\approx 0.004$ (orange curve). (b) The standard deviation of $\rho^{(1)}$ eigenvalue distribution for random states with $N=4$, as a function of the single-particle dimension $D$. The numerical results obtained over 1000 realizations for each $D$ up to $D=50$ are shown with green triangles. The orange dots show the standard deviations of the corresponding semi-circular distributions [Eq.~(\ref{eq:mu and sigma trace fixe SC from GSUE})]. The blue crosses show the standard deviations of numerically obtained eigenvalue distributions for trace-fixed WL ensembles over 1000 realizations for each $D$. The $\rho^{(1)}$ numerical result approaches the analytical semi-circular and numerical WL results in the large-$D$ limit. For $D\gg N$, the standard deviation of the semi-circular distribution has an asymptotic form $\frac{N}{D}\sqrt{D/\binom{D}{N-1}} \overset{D\gg N}{\sim} N\sqrt{(N-1)!}D^{-N/2}=9.78/D^2$. The dashed curve is obtained by fitting $\alpha/D^2$ to the $D=50$ green triangle, which yields $\alpha=9.95$. 
    }
    \label{fig:SC distribution and std deviation as a function of D}
\end{figure}

\noindent where $\sigma = \frac{\text{Tr}\rho^{(M)}}{n}\sqrt{c}$. This is a semi-circular distribution of the form of Eq.~(\ref{eq:WL SC distribution}), with zero mean and a standard deviation given by $\sigma$. We note that the condition $||\delta W|| \sim \sigma \ll \text{Tr}\rho^{(M)}/n$ holds in the limit of $c=n/m \rightarrow 0$.

As stated earlier, to obtain the spectral density for the eigenvalues of the matrix $W$, we simply shift the eigenvalues of $\delta W$ by $\text{Tr}\rho^{(M)}/n$, which, in turn, shifts the entire semi-circular distribution Eq.~(\ref{eq:Eigenvalue x SC}) by the same amount. Therefore, the spectral density for the trace-fixed WL ensemble in the $n\rightarrow\infty$ and $c\rightarrow 0$ limit is given by
\begin{equation} \label{eq:Eigenvalue x SC final form}
\begin{aligned}
P_{\text{SC}}(z) = \frac{1}{2 \pi \sigma^2} \sqrt{4\sigma^2 - (z-\mu)^2},
\end{aligned}
\end{equation}
where
\begin{equation} \label{eq:mu and sigma trace fixe SC from GSUE}
\begin{aligned}
\mu = \frac{\text{Tr}\rho^{(M)}}{n} \,\, \text{and} \,\, \sigma = \frac{\text{Tr}\rho^{(M)}}{n}\sqrt{c}.
\end{aligned}
\end{equation}

We find that the eigenvalue distribution of the ensemble of DMs $\rho^{(M)}$ matches the spectral density of the trace-fixed WL ensemble when $n=\binom{D}{M} \gg 1$ and $c=\binom{D}{M}/\binom{D}{N-M} \ll 1$ [Eq.~(\ref{eq:Eigenvalue x SC final form})], see Fig.~\ref{fig:SC distribution and std deviation as a function of D}. In our numerical computation for the DM ensemble, each realization of a random $N$-fermion state with single-particle dimension $D$ is generated by randomly assigning its $\binom{D}{N}$ components in a fixed set of $\binom{D}{N}$ SD basis states. The real and imaginary parts of the components are independently drawn from the Gaussian distribution $N(0,1)$ and the state is normalized afterward. The DM $\rho^{(M)}$ for each realization is then computed using Eqs.~(\ref{eq:state}), (\ref{eq:bipartite representation}), and (\ref{eq:gamma gamma dagger}).

\subsection{Trace-fixed WL ensemble for general \texorpdfstring{$c$}{c}}\label{subsec:trace fixing of spectral density}

Interestingly, if we fix the trace of the Wishart matrix $W$ to $\text{Tr}\rho^{(M)}$ at the level of Mar\v{c}enko-Pastur scaling function in Eq.~(\ref{eq:WL spectral density distribution}), we can obtain the spectral density of the trace-fixed WL ensemble in the $n\rightarrow\infty$ limit for a general value of the parameter $c=n/m$ in the range $[0,1]$. In the $c\rightarrow 0$ limit, this distribution reduces to the semi-circular distribution Eq.~(\ref{eq:Eigenvalue x SC final form}) with the same $\mu,\sigma$ as in Eq.~(\ref{eq:mu and sigma trace fixe SC from GSUE}).

We consider the expectation value of the trace of the Wishart matrix, $\text{Tr}W = \sum_{i=1}^{n}x_i$, in the large-$n$ limit, where $x_i$-s are the $n$ eigenvalues of $W$. It is known that the expectation value of any \textit{linear statistics} [i.e. quantity of the form $\textit{L} = \sum_{i=1}^n f(x_i)]$ can be computed as a one-dimensional integral in terms of the spectral density of the ensemble as \cite{LivanVivo2011} 
\begin{equation} \label{eq:Linear statistics expextation value equation}
\begin{aligned}
\left\langle \sum_{i=1}^n f(x_i) \right\rangle_{\text{WL}} = n \int^{\mu_{\text{WL}} +2\sigma_{\text{WL}}}_{\mu_{\text{WL}} -2\sigma_{\text{WL}}}  f(x) P_{\text{WL}}(x)|_{n\rightarrow\infty} dx,
\end{aligned}
\end{equation}
where $P_{\text{WL}}(x)|_{n\rightarrow\infty
}$ is given by Eq.~(\ref{eq:WL spectral density distribution}), and the integral is taken over its domain.
Then, for $\text{Tr}W$, we have 
\begin{equation} \label{eq:trace fixing condition}
\begin{aligned}
\langle \text{Tr}W \rangle_{\text{WL}} = n \int^{\mu_{\text{WL}} +2\sigma_{\text{Wl}}}_{\mu_{\text{WL}} -2\sigma_{\text{Wl}}} x P_{\text{WL}}(x)|_{n\rightarrow\infty} dx.
\end{aligned}
\end{equation}
In order to fix the trace to the desired value of $\text{Tr}\rho^{(M)}$, we divide by $\langle \text{Tr}W \rangle_{\text{WL}}$ and multiply by $\text{Tr}\rho^{(M)}$ on both sides of Eq.~(\ref{eq:trace fixing condition}). We then perform a change of variable to $ z = \frac{\text{Tr}\rho^{(M)}}{\langle \text{Tr}W \rangle_{\text{WL}}} x$. This transformation results in a trace-fixed distribution $P_{\text{TWL}}(z)$ for the variable $z$ and has the form
\begin{equation} \label{eq:TWL derivation}
\begin{aligned}
P_{\text{TWL}}(z) = \frac{\langle \text{Tr}W \rangle_{\text{WL}}}{\text{Tr}\rho^{(M)}} P_{\text{WL}}\left(\frac{\langle \text{Tr}W \rangle_{\text{WL}}}{\text{Tr}\rho^{(M)}} z \right)_{n\rightarrow\infty}.
\end{aligned}
\end{equation}
Here, $\langle \text{Tr}W \rangle_{\text{WL}}$ is given by Eq.~(\ref{eq:trace fixing condition}), whose right-hand side is just $n\mu_\text{WL}$. Using $\mu_{\text{WL}} = 2n/c$ from Eq.~(\ref{eq:WL large n,m mean and std dev}), we find that $\langle \text{Tr}W \rangle_{\text{WL}} = 2n^2/c$. Plugging this into Eq.~(\ref{eq:TWL derivation}) and simplifying, we obtain the spectral density of the trace-fixed WL ensemble in the $n\rightarrow\infty$ limit for a general value of $c=n/m\in[0,1]$
\begin{equation} \label{eq:TWL derivation final result}
\begin{aligned}
P_{\text{TWL}}(z) = \frac{1}{2 \pi z} \sqrt{\left(\frac{z}{\mu c} - \xi_{-}\right)\left(\xi_{+} - \frac{z}{\mu c} \right)},
\end{aligned}
\end{equation}
where again $\mu = \text{Tr}\rho^{(M)}/n$. This is just the rescaled Mar\v{c}enko-Pastur distribution. The mean value and the standard deviation of this distribution match those in Eq.~(\ref{eq:mu and sigma trace fixe SC from GSUE}). Taking the $c \rightarrow 0$ limit, we recover the semi-circular distribution given by Eq.~(\ref{eq:Eigenvalue x SC final form}). 

We find that the eigenvalue distribution of the ensemble of DMs $\rho^{(M)}$ matches the spectral density of the trace-fixed WL ensemble when $n \gg 1$ for general values of $c\in[0,1]$ [Eq.~(\ref{eq:TWL derivation final result})], see Fig.~\ref{fig:plot2 N=5, M=2} for $c=1/6$ and Fig.~\ref{fig:plot1 N=4, M=2} for $c=1$.

\subsection{Many-body entanglement of  random states}\label{subsec:entanglement of random state}

We have demonstrated that the eigenvalue distribution of the DMs, $\rho^{(M)}$, of random $N$-fermion states with single-particle dimension $D$ matches the spectral density of the trace-fixed WL ensemble with parameters $n=\binom{D}{M}$ and $ m=\binom{D}{N-M}$ ($n\leq m$), in the $n\rightarrow\infty$ limit. We now translate the condition $n\leq m$ and the large-$n$ limit into the DM parameters $D,N$, and $M$. As mentioned in Sec.~\ref{subsec:M-body DM}, we consider the case where $M\leq N/2$, for which the condition $n\leq m \Leftrightarrow \binom{D}{M} \leq \binom{D}{N-M}$ is satisfied. In particular, $M=N/2$ corresponds to $c=n/m=1$, and $M<N/2$ corresponds to $c\in[0,1)$. The $n=\binom{D}{M}\rightarrow \infty$ limit is equivalent to $D \gg 1$. We have also considered the limit where $c\rightarrow 0$, which can be shown to correspond to $D \gg N>2M$.

In the large-$D$ and $D\gg N$ limit, the $M$-body DMs of random $N$-fermion states have a semi-circular eigenvalue distribution given by Eq.~(\ref{eq:Eigenvalue x SC final form}) when $M<N/2$, since $c\rightarrow 0$ in this limit. This occurs for odd $N$ with any allowed value of $M$, in which case $M$ is strictly less than $N/2$, and for even $N$ when $M < N/2$. As $c \rightarrow 0$, we see from the form of the mean and the standard deviation in Eq.~(\ref{eq:mu and sigma trace fixe SC from GSUE}) that $\sigma$ decreases much faster than $\mu$, causing the semi-circular distribution to become increasingly peaked around $\mu$. Therefore, in the large-$D$ and $D \gg N$ limit, the DMs of random states can be approximated as $\rho^{(M)} \approx \frac{\text{Tr}\rho^{(M)}}{n} \mathbb{I}_{n \times n}$ for all $M<N/2$, corresponding to the maximum entanglement entropy. 

In contrast, the $M$-body DMs exhibit a qualitatively different eigenvalue distribution when $N$ is even and $M=N/2$ (i.e., $c=1$),
\begin{equation} \label{eq:TWL distribution c=1 case}
\begin{aligned}
P_{\text{TWL}}(z)|_{c=1} = \frac{1}{2 \pi \mu z} \sqrt{z\left(4\mu - z\right)},
\end{aligned}
\end{equation}
which follows from substituting $\xi_-|_{c=1}=0$ and $\xi_+|_{c=1} = 4$ into Eq.~(\ref{eq:TWL derivation final result}). The plot in Fig.~\ref{fig:plot1 N=4, M=2} shows the eigenvalue distribution of $\rho^{(N/2)}$ corresponding to $c=1$, which is dramatically different from Fig.~\ref{fig:plot2 N=5, M=2} corresponding to $c < 1$. From the perspective of the DMs, this change occurs because when $M=N/2$, only Steiner systems correspond to states with maximum $\frac{N}{2}$-body entanglement entropy [see Sec.~\ref{subsec:existence of maximally entangled states for given D,N} and the yellow lines in Fig.~\ref{fig:summary D-N plane}]. This remains true for any value of $D$, thus heavily limiting the set of maximally $\frac{N}{2}$-body entangled states, resulting in a drastically different eigenvalue distribution. 

We have established in Sec.~\ref{subsec:existence of maximally entangled states for given D,N} that an $N$-fermion state cannot be maximally $M$-body entangled for $M> N/2$, and maximum $M$-body entanglement is only possible for $M\leq N/2$. Meanwhile, the $M$-body entanglement satisfies a \hyperref[nesting-theorem]{nesting property} that if an $N$-fermion state is maximally $M$-body entangled, it is also $M'$-body entangled for any $M'\leq M$. These two facts suggest that we define a notion of absolute maximum entanglement. An $N$-fermion state is  \textit{absolutely maximally entangled} if it is $M$-body entangled for all $M<N/2$ when $N$ is odd, and for all $M\leq N/2$ when $N$ is even. As a result, $N$-fermion random states with odd $N$ are absolutely maximally entangled in the large-$D$ and $D \gg N$ limit. On the other hand, for $N$-fermion random states with even $N$, the $(N/2)$-body DMs exhibit a broad distribution rather than a peaked one, indicating that they are not maximally $(N/2)$-body entangled, hence not absolutely maximally entangled. Later we will show that the mean $(N/2)$-body entanglement entropy of random states is almost maximum, up to an additive negative constant that does not depend on $D$. This means that random states with even $N$ are ``nearly'' absolutely maximally entangled in the large-$D$ and $D \gg N$ limit.

As an interesting observation, it is known that by performing a change of variables $s = \sqrt{z}$ in Eq.~(\ref{eq:TWL distribution c=1 case}), the distribution for the variable $s$ becomes proportional to the semi-circular distribution: $P_{\text{TWL}}(s)|_{c=1} = \frac{2}{2 \pi \mu} \sqrt{4\mu - s^2}$. It is centered around zero and $\mu=\frac{\text{Tr}\rho^{(M)}}{n}$ serves as the standard deviation in this case. Thus, in the case of $c=1$, the square roots of the eigenvalues of the corresponding DMs are distributed according to the semi-circular law, rather than the eigenvalues themselves \cite{PottersBouchaud2020}.

\begin{figure}[!h]
    \centering
    \subfigure[]{
        \includegraphics[width=1.0\columnwidth]{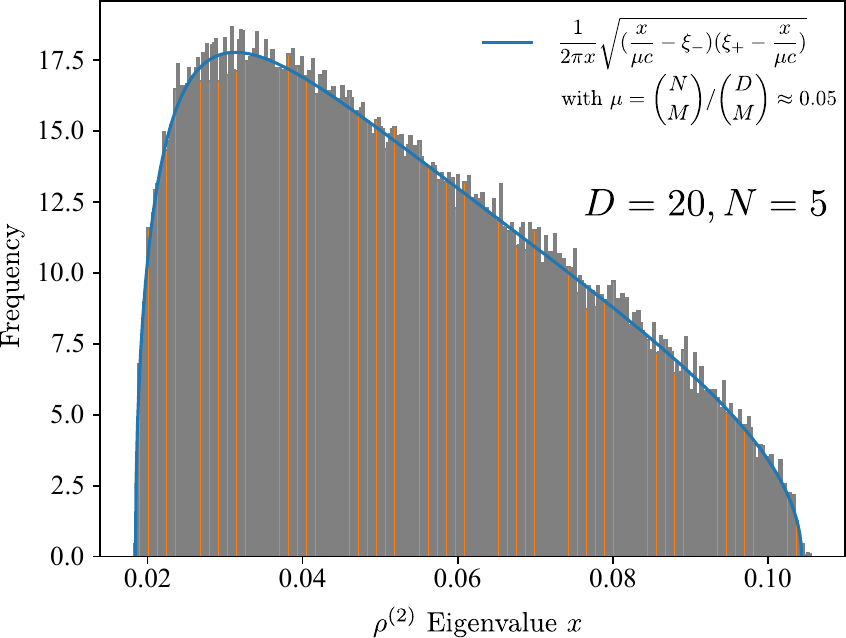}
        \label{fig:plot2 N=5, M=2}
    }
    \subfigure[]{
        \includegraphics[width=1.0\columnwidth]{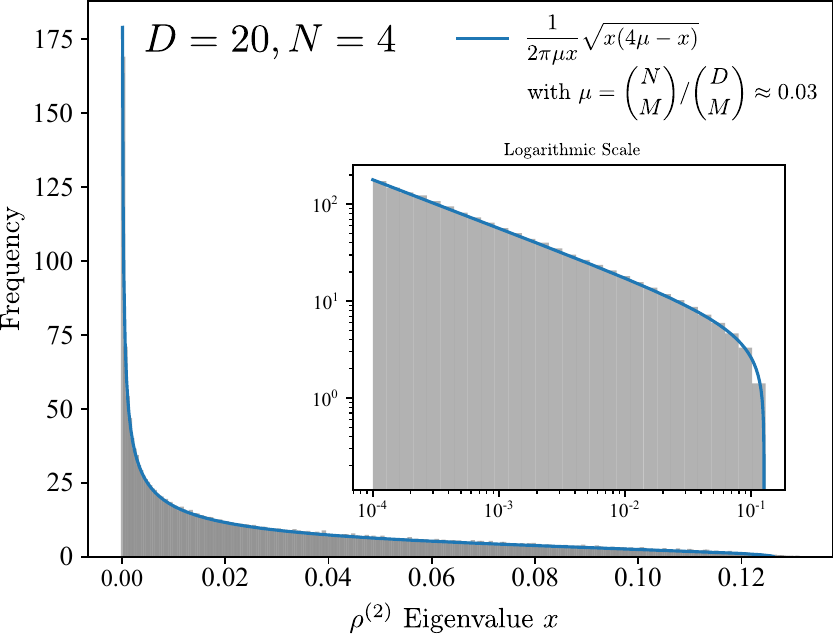}
        \label{fig:plot1 N=4, M=2}    
    }
    \caption{Eigenvalue distributions of two-body DMs $\rho^{(2)}$ for randomly generated states with $D=20$, in the cases $N=5$ and $N=4$. In both cases, the numerical results (gray histograms) obtained over 400 realizations are well approximated by the analytically derived eigenvalue distributions of the trace-fixed WL ensembles (blue curves). (a)~The case of $N=5$. The analytical WL curve [Eq.~(\ref{eq:TWL derivation final result})] has parameters $c=\binom{D}{M}/\binom{D}{N-M}=1/6$, $\mu=\binom{N}{M}/\binom{D}{M}\approx 0.05$, and $\xi_\pm=(1\pm1/\sqrt{c})^2=(1\pm\sqrt{6})^2$. The distribution is peaked around an eigenvalue $\mu$ corresponding to a maximally entangled state. (b)~The case of $N=4$. This is the case when $M=N/2$ (equal bipartitioning), or equivalently $c=1$. The analytical WL curve is given by Eq.~(\ref{eq:TWL distribution c=1 case}) with $\mu\approx 0.03$. The inset shows the plot on a logarithmic scale. The distribution is broad, which is qualitatively different from the distribution in (a), where $c\ll 1$ (or $D \gg N > 2M$). We relate this difference to the fact that for $M=N/2$, maximum entanglement entropy is achieved only by states associated with Steiner systems (Sec.~\ref{subsec:existence of maximally entangled states for given D,N}), making it a rare occurrence. In contrast, when $M<N/2$, most states are maximally $M$-body entangled in the large-$D$ limit. 
    }
    \label{fig:TWL general distribution for N=4,5 M=2 D=20}
\end{figure}

\begin{figure*}[t]
    \centering
    \subfigure[]{
        \includegraphics[width=1\columnwidth]{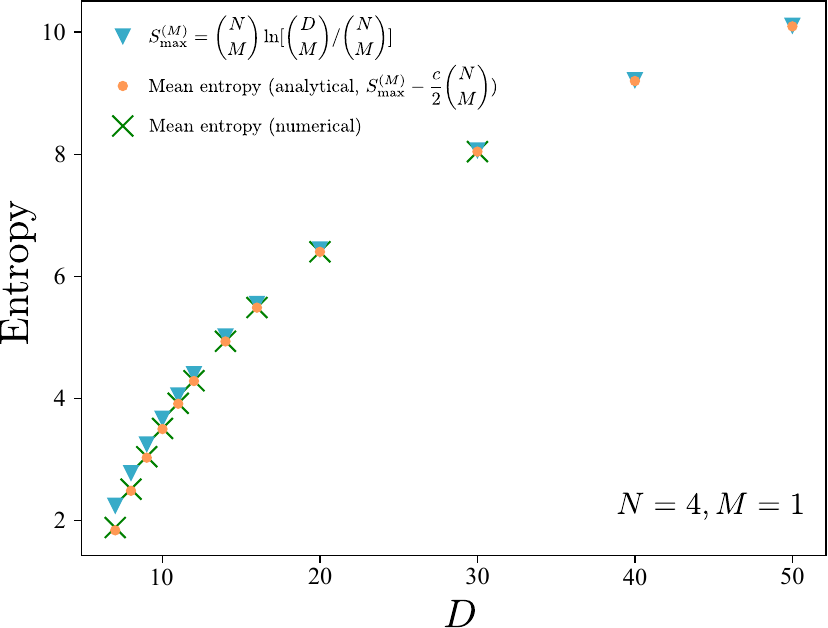}
        \label{fig:plot1 N=4, M=1 mean entropy}
    }
    \hfill
    \subfigure[]{
        \includegraphics[width=1\columnwidth]{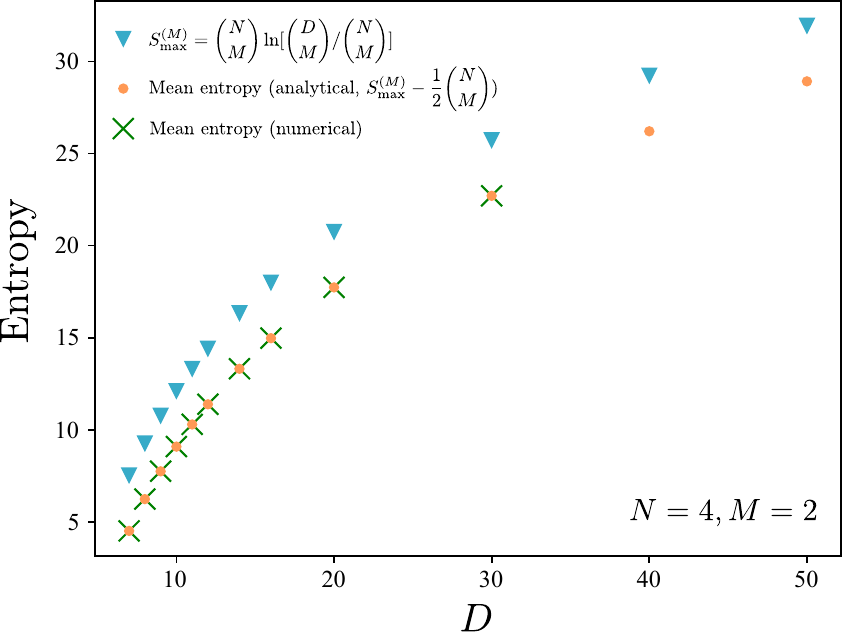}
        \label{fig:plot2 N=4, M=2 mean entropy}
    }
    \caption{Mean $M$-body entanglement entropy of randomly generated states with $N=4$ as a function of the single-particle dimension $D$, for $M=1$ and $M=2$. In both cases, the numerical results (green crosses), obtained over $10^4$ realizations for each $D$ up to $D=20$ and $10^3$ realizations for $D=30$, are well approximated by the mean entropies analytically derived from the trace-fixed WL ensembles (orange dots). The maximum entanglement entropies for each $D$ are shown as blue triangles. (a)~The case of $M=1$. The analytical WL mean entropies are given by Eq.~(\ref{eq:expectation value of entropy S over SC distr}). In the large-$D$ limit, the parameter $c=\binom{D}{M}/\binom{D}{N-M}=\frac{6}{(D-1)(D-2)}$ goes to 0 and the mean entropy approaches the maximum entropy. (b)~The case of $M=2$. This is the case of $M=N/2$ (equal bipartitioning), or equivalently $c=1$. The analytical WL mean entropies are given by Eq.~(\ref{eq:expectation value of entropy S over c=1 distr}), which are less than the maximum entropies by a constant value of $\frac{1}{2}\binom{N}{M}=3$ for all values of $D$. This asymptotic behavior of the mean entropy is qualitatively different from the behavior in (a). We relate this difference to the fact that for $M=N/2$, maximum entanglement entropy is achieved only by states associated with Steiner systems [see Sec.~\ref{subsec:existence of maximally entangled states for given D,N} and Fig.~\ref{fig:plot1 N=4, M=2} caption] and is therefore a rare occurrence. As a result, the random states do not become maximally $(N/2)$-body entangled in the large-$D$ limit.
    }
    \label{fig: mean entropy vs D for N=4 and M=1, M=2}
\end{figure*}

Next, we calculate the expectation value of the entropy using the trace-fixed WL ensemble in the large-$D$ and $D \gg N$ limit. We note that the entropy $S = -\sum_{i=1}^{n}{z_i \ln{z_i}}$ is a linear statistics, so its expectation value can be written in terms of the the spectral density $P_{\text{TWL}}(z)$ using Eq.~(\ref{eq:Linear statistics expextation value equation}) as
 
\begin{equation} \label{eq:expectation value of entropy S}
\begin{aligned}
\langle S \rangle_{\text{TWL}} = -n \int_{\mu c \xi_{-}}^{\mu c \xi_{+}} dz (z \ln{z}) P_{\text{TWL}}(z),
\end{aligned}
\end{equation}
where integration is over the domain of distribution $P_{\text{TWL}}(z)$ in Eq.~(\ref{eq:TWL derivation final result}). As discussed above, in the case when $D\gg N$, two distinct scenarios arise: $c=1$ and $c \rightarrow 0$, which we address separately.

\textit{Case of $c \rightarrow 0$}: This corresponds to the case when $D \gg N >2M$. As was pointed out, In this limit $P_{\text{TWL}}(z)$ approaches the semi-circular distribution. Using the semi-circular distribution from Eq.~(\ref{eq:Eigenvalue x SC final form}) in Eq.~(\ref{eq:expectation value of entropy S}) we can approximate the integral by performing a change of variables $z = \mu + x$, where $-2\sigma \leq x \leq 2\sigma$. As $c \to 0$, this implies $\sigma \ll \mu$, and consequently, $x \ll \mu$. Expanding $z\ln{z}$ in powers of $x$ and retaining terms up to quadratic order, we obtain
\begin{equation} \label{eq:expension of zlnz}
\begin{aligned}
z\ln{z} \approx \mu \ln{\mu} + x(\ln{\mu} + 1) + \frac{1}{2}\frac{x^2}{\mu}.
\end{aligned}
\end{equation}
The first term is constant and, due to the normalization of the distribution, integrates to $-n\mu\ln{\mu}$. The second term becomes zero since it is proportion to the mean $\langle x \rangle_{\text{TWL}}$ which we shifted to zero. Finally, the quadratic term in Eq.~(\ref{eq:expension of zlnz}) integrates to a quantity proportional to $\langle x^2 \rangle_{\text{TWL}} = \sigma^2$. Overall, the approximate integration gives us
\begin{equation} \label{eq:expectation value of entropy S over SC distr}
\begin{aligned}
\langle S \rangle_{\text{TWL}}  \overset{c \rightarrow 0}{\approx}  S^{(M)}_\text{max} - \frac{\text{Tr}\rho^{(M)}}{2} c,
\end{aligned}
\end{equation}
where we have inserted the explicit expressions for $\mu$ and $\sigma$ from Eq.~(\ref{eq:mu and sigma trace fixe SC from GSUE}). We see that, in the $c \to 0$ limit, the mean value of the entropy very rapidly approaches the maximum entropy value [see Eq.~(\ref{eq:S maximum entropy general}) for $S_{\text{max}}^{(M)}$]. Fig.~\ref{fig:plot1 N=4, M=1 mean entropy} shows the dependence of the mean entropy on $D$ for the case when $N=4$ and $M=1$.

\textit{Case of $c=1$}: This is the case of equal bipartitioning, when $M=N/2$. Using the form of the distribution given in Eq.~(\ref{eq:TWL distribution c=1 case}) and performing a change of variables $z=4\mu x$, the probability distribution for $x$ becomes $P_{\text{TWL}}(x) = \frac{4}{2\pi x}\sqrt{x(1-x)}$, where $x \in (0,1)$. At the same time, we have $z \ln{z} = 4\mu\ln{(4\mu)} x + 4\mu x\ln{x}$. The mean value of the entropy in Eq.~(\ref{eq:expectation value of entropy S}) then reads 
\begin{equation}\label{eq:expectation value of entropy S over c=1 distr, intermediate result}
\begin{aligned}
\langle S \rangle_{\text{TWL}} = \frac{-n (4\mu)^2 }{\mu} (a_1 \ln{4\mu} + a_2),
\end{aligned}
\end{equation} 
where $a_1 = \frac{1}{4}\int_0^1 dx \, x P_\text{TWL}(x) = \frac{1}{4} \langle x \rangle_{\text{TWL}} = \frac{1}{16}$ and $a_2 = \frac{1}{4} \int_0^1 dx \, x\ln{x} P_{\text{TWL}}(x) \approx -0.055393 \approx \frac{1}{16}(1/2 - \ln{4})$, with $a_2$ being numerically integrated. In evaluating $a_1$, we used the fact that $\langle x \rangle_{\text{TWL}} = 1/4$ since $\mu \equiv \langle z \rangle_{\text{TWL}} = 4\mu \langle x \rangle_{\text{TWL}}$, where expectation values of $z$ and $x$ are taken using their respective probability distributions.

Plugging everything back into Eq.~(\ref{eq:expectation value of entropy S over c=1 distr, intermediate result}) gives us
\begin{equation} \label{eq:expectation value of entropy S over c=1 distr}
\begin{aligned}
\langle S \rangle_{\text{TWL}}  \overset{c=1}{\approx}  S^{(\frac{N}{2})}_\text{max} - \frac{\text{Tr}\rho^{(\frac{N}{2})}}{2} .
\end{aligned}
\end{equation}
The mean entropy always stays lower than the maximum entropy by a constant value of $\text{Tr}\rho^{(\frac{N}{2})}/2$. Again, we relate this to the fact that when $M=N/2$, the maximum entanglement entropy is achieved only by states associated with Steiner systems [see Sec.~\ref{subsec:existence of maximally entangled states for given D,N} and Fig.~\ref{fig:summary D-N plane}]. Consequently, a random $N$-fermion state is not maximally $\frac{N}{2}$-body entangled, preventing the mean entropy from reaching the maximum value. Nevertheless, since $S^{(M)}_\text{max}$ grows as $\text{Tr}\rho^{(M)} \ln{n} = \text{Tr}\rho^{(M)} \ln{\binom{D}{M}}$ [see Eq.~(\ref{eq:S maximum entropy general})], even in this case, the mean entropy is very close to the maximum entropy for large values of $D$.  Fig.~\ref{fig:plot2 N=4, M=2 mean entropy} shows the dependence of the mean entropy on $D$ when $c=1$, with $N=4$ and $M=2$.

\section{Summary and outlook}\label{sec:summary and outlook}

We characterized maximally many-body entangled $N$-particle fermionic states by using $M$-body reduced density matrices, which provide a basis-independent way to characterize quantum correlations between subsets of particles. 
Of particular interest are the 1-body and 2-body density matrices that sufficiently characterize Hamiltonians with 2-body interactions and typical linear response functions. 
We established a connection between fermionic states and the mathematical structure of hypergraphs. Specifically, we showed that $t$-designs correspond to maximally $M$-body entangled fermionic states when $t=M$. In general, the existence of maximally $M$-body entangled $N$-fermion states for given single-particle dimension $D$ depends on the solvability of systems of equations of the form Eq.~(\ref{eq:System of equations general M}). 

We used a numerical search (that can be limited to hypergraph isomorphism classes) to find maximally $M$-body entangled states for relatively small $D$ and $N$. 
Figure~\ref{fig:D and N mesh plot for M=1 and M=2} presents the results of this search for $M=1$ and $M=2$. 

Another major focus of this work are the statistical properties of the $M$-body DMs and the associated entanglement entropy for random fermionic states in the large-$D$ limit. For $D\gg N > 2M$, 
random fermionic states become maximally $M$-body entangled with the eigenvalue distribution of the $M$-body DM having semi-circular shape, Eq.~(\ref{eq:Eigenvalue x SC final form}), sharply peaked around the eigenvalue corresponding to a maximally mixed $\rho^{(M)}$. 
This eigenvalue distribution matches the predictions derived from the trace-fixed Wishart-Laguerre random matrix ensemble, see Fig.~\ref{fig:plot1 SC distribution D=50}. 

The case of equal bipartitioning, $M=N/2$, results in a qualitatively different eigenvalue distribution given by Eq.~(\ref{eq:TWL distribution c=1 case}), see also Fig.~\ref{fig:plot1 N=4, M=2}. 
In addition, $\frac{N}{2}$-body entanglement entropy does not converge to maximum entropy value, as shown in Eq.~(\ref{eq:expectation value of entropy S over c=1 distr}) and Fig.~\ref{fig:plot2 N=4, M=2 mean entropy}. This occurs because when $M=N/2$, only Steiner systems (a special case of a hypergraph $t$-design, see Sec.~\ref{subsec:mini-review on hypergraphs}) correspond to maximally $\frac{N}{2}$-body entangled fermionic states.

This paper numerically establishes that the statistical properties of $M$-body DMs are well described by the trace-fixed WL ensemble, as one might expect from Eq.~(\ref{eq:gamma gamma dagger}). However the rigorous proof of the connection between the ensemble of $M$-body DMs derived from random fermionic states and the trace-fixed WL ensemble has not been formally established. Investigating this connection would provide a deeper understanding of why the trace-fixed WL ensemble effectively captures the eigenvalue statistics of $M$-body DMs.

As an outlook on the question of fermionic quantum computation applied to quantum chemistry and the many-electron problem, our work opens multiple new directions. 
For example, in the context of maximally entangled states one may ask how such states can be created with fermionic quantum circuits by starting from an unentangled single Slater determinant state. 
This question of circuit complexity is also related to quantifying the $M$-body entanglement properties of ground states of local and non-local interacting Hamiltonians, which could lead to an entanglement classification of many-fermion systems.

In the context of bootstrapping techniques for strongly correlated electron systems, which aim to obtain the ground state energy by minimizing a functional of the $M$-body DM~\cite{GaoKhalaf2024}, our statistical findings on $M$-body DMs could provide valuable insights. Specifically, these findings may enable to rephrase the optimization problem with $N$-representability constraints as an approximate optimization over the trace-fixed WL ensemble. Such a reformulation has the potential to significantly accelerate numerical computations and further motivates better understanding of the connection between ensemble of $M$-body DMs and trace-fixed WL ensemble.

\acknowledgements

\section{Acknowledgements}\label{sec:acknowledgements}
It is a pleasure to thank Teemu Ojanen, Leo Zhou, Marius Junge, Eslam Khalaf, for useful discussions.
JIV thanks the Max Planck Institute for Solid State Research for hospitality. 
JIV is especially grateful to Teemu Ojanen for directing his attention to complexity of fermionic states and Ref.~\cite{VanhalaOjanen2024}. Support for this research was provided by the Office of
the Vice Chancellor for Research and Graduate Education at the University of Wisconsin–Madison with funding from the Wisconsin Alumni Research Foundation.
This research was supported in part by grant NSF PHY-
2309135 to the Kavli Institute for Theoretical Physics
(KITP). EJK acknowledges hospitality by the KITP.

\appendix

\section{Subadditivity and strong subadditivity for maximally \texorpdfstring{$M$}{M}-body entangled states}\label{sec:appendix subadditivity and strong subadditivity}

Here we show that for a maximally $M$-body entangled state, the subadditivity inequality Eq.~(\ref{eq:subadditivity}) and the strong subadditivity inequality Eq.~(\ref{eq:strong subadditivity}) are satisfied. Specifically, let $\ket{\Psi}$ be an $N$-fermion state of single-particle dimension $D$. If $\ket{\Psi}$ is maximally $(M_1+M_2)$-body entangled, i.e., if $\rho_\Psi^{(M_1+M_2)} = \frac{\binom{N}{M_1+M_2}}{\binom{D}{M_1+M_2}} \mathbb{I}$, then
\begin{equation}\label{eq:subadditivity in Appendix}
\begin{aligned}
    S(\rho_{n,\Psi}^{(M_1 + M_2)}) \leq S(\rho_{n,\Psi}^{(M_1)}) + S(\rho_{n,\Psi}^{(M_2)}).
\end{aligned}
\end{equation}
If $\ket{\Psi}$ is maximally $(M_1+M_2+M_3)$-body entangled, i.e., if $\rho_\Psi^{(M_1+M_2+M_3)} = \frac{\binom{N}{M_1+M_2+M_3}}{\binom{D}{M_1+M_2+M_3}} \mathbb{I}$, then
\begin{equation}\label{eq:strong subadditivity in Appendix}
\begin{aligned}
    S(\rho_{n,\Psi}^{(M_1 + M_2 + M_3)}) \leq & \,\, S(\rho_{n,\Psi}^{(M_1 + M_3)}) + S(\rho_{n,\Psi}^{(M_2 + M_3)})\\
    & - S(\rho_{n,\Psi}^{(M_3)}).
\end{aligned}
\end{equation}

We first show Eq.~(\ref{eq:subadditivity in Appendix}). In Appendix \ref{sec:appendix proof of theorem M leads to M'}, we prove that if a state is maximally $M$-body entangled, it is maximally $M'$-body entangled for any $M'\leq M$. Since $M_1,M_2\leq M_1+M_2$, a maximally $(M_1+M_2)$-body entangled state $\ket{\Psi}$ is also maximally $M_1$ and $M_2$-body entangled. Using Eqs.~(\ref{eq:normalized DM entropy}) and~(\ref{eq:S maximum entropy general}), we get the entropy values for the normalized DMs
\begin{align}
    S(\rho_{n,\Psi}^{(M_1+M_2)}) &= \ln{\binom{D}{M_1+M_2}}\\
    S(\rho_{n,\Psi}^{(M_1)}) &= \ln{\binom{D}{M_1}} \\
    S(\rho_{n,\Psi}^{(M_2)}) &= \ln{\binom{D}{M_2}}.
\end{align}
The subadditivity inequality Eq.~(\ref{eq:subadditivity in Appendix}) now reads
\begin{equation}\label{eq:subadditivity proof 1 in appendix}
\binom{D}{M_1+M_2} \leq \binom{D}{M_1} \binom{D}{M_2}.
\end{equation}
This can be equivalently written as
\begin{equation}\label{eq:subadditivity proof 2 in appendix}
\frac{M_1! M_2!}{(M_1+M_2)!} \leq \frac{D!(D'-M_2)!}{(D-M_2)! D'!},
\end{equation}
with $D' = D - M_1$, or
\begin{equation}\label{eq:subadditivity rewriting in appendix}
    \frac{M_2(M_2-1)\cdots 1}{(M_1+M_2)\cdots (M_1+1)} \leq \frac{D(D-1)\cdots(D-M_2+1)}{D'(D'-1)\cdots(D'-M_2+1)}.
\end{equation}
On the left-hand side (LHS) of Eq.~(\ref{eq:subadditivity rewriting in appendix}), each factor in the numerator is less than a corresponding factor in the denominator, so $\text{LHS}<1$. On the right-hand side (RHS), since $D>D'=D-M_1$, each factor in the numerator is greater than a corresponding factor in the denominator, leading to $\text{RHS}>1$. Therefore, the inequality Eq.~(\ref{eq:subadditivity rewriting in appendix}) and equivalently Eq.~(\ref{eq:subadditivity in Appendix}) are satisfied.

The Eq.~(\ref{eq:strong subadditivity in Appendix}) can be shown in a similar way. A maximally $(M_1+M_2+M_3)$-body entangled state $\ket{\Psi}$ is also $M_1+M_3$, $M_2+M_3$, and $M_3$-body entangled, and Eq.~(\ref{eq:strong subadditivity in Appendix}) reads
\begin{equation}\label{eq:strong subadditivity proof 1 in appendix}
\binom{D}{M_1+M_2+M_3} \leq \frac{\binom{D}{M_1+M_3} \binom{D}{M_2+M_3}}{\binom{D}{M_3}}.
\end{equation}
This can be equivalently written as
\begin{equation}\label{eq:strong subadditivity proof 2 in appendix}
\frac{(M')!(M_3+M_2)!}{(M'+M_2)!(M_3)!} \leq \frac{\tilde{D}!(D''-M_1)!}{(\tilde{D}-M_1)!D''!}, 
\end{equation}
with $D''=D-M_2-M_3$, $\tilde{D} = D - M_3$, and $M' = M_1 + M_3$. Again, the LHS of Eq.~(\ref{eq:strong subadditivity proof 2 in appendix}) is less than 1 due to $M'>M_3$ and the RHS is greater than 1 due to $\tilde{D}>D''$. Therefore, the inequality Eq.~(\ref{eq:strong subadditivity proof 2 in appendix}) and equivalently Eq.~(\ref{eq:strong subadditivity in Appendix}) are satisfied.

\section{\texorpdfstring{$M$}{M}-designs give maximally \texorpdfstring{$M$}{M}-body entangled states}
\label{sec:appendix M-designs give maximally entangled states}

In this section, we show that $M$-designs satisfying the \hyperref[overlap-criterion]{overlap criterion} give rise to maximally $M$-body entangled states. As a reminder, in Sec.~\ref{subsec:mini-review on hypergraphs}, we introduce the language of hypergraphs and establish a one-to-one mapping between sets of $N$-fermion SDs and $N$-uniform hypergraphs. For a given set of $N$-fermion SDs $\{\ket{\text{SD}_k}\}_{k=1}^b$ of single-particle dimension $D$, or equivalently an $N$-uniform hypergraph of $D$ vertices, one considers Eq.~(\ref{eq:System of equations general M})
\begin{equation}\label{eq:Appendix system of equations general M}
    A^{(M)}\vec{x} = \frac{\binom{N}{M}}{\binom{D}{M}}\vec{1}_{\binom{D}{M}},
\end{equation}
where $A^{(M)}$ is the $M^\text{th}$ incidence matrix of the hypergraph. If the hypergraph satisfies the overlap criterion that any two edges share fewer than $N-M$ vertices, and if Eq.~(\ref{eq:Appendix system of equations general M}) admits a solution $\vec{x}$ with all entries being non-negative and summing to 1, then maximally $M$-body entangled states $\ket{\Psi}$ can be constructed from the given set of SDs. Specifically, the states given by $\ket{\Psi}=\sum_{k=1}^b \alpha_k\ket{\text{SD}_k}$, where $|\alpha_k|^2=x_k$, are maximally $M$-body entangled.

Here we prove that when the given $N$-uniform hypergraph is an $M$-design, Eq.~(\ref{eq:Appendix system of equations general M}) has a solution $\vec{x}=\frac{1}{b}\vec{1}_b$, where $b$ is the number of edges in the hypergraph. Consequently, any $M$-design satisfying the \hyperref[overlap-criterion]{overlap criterion} leads to maximally $M$-body entangled states with coefficients $\alpha_k$ all equal to $1/\sqrt{b}$ up to a phase. We also show the converse that if $\vec{x}=\frac{1}{b}\vec{1}_b$ is a solution to Eq.~(\ref{eq:Appendix system of equations general M}), then the equation must come from an $M$-design.

\begin{theorem}[\textbf{$M$-design}]\label{theorem:M-design}
Let $(V,E)$ be an N-uniform hypergraph with $D$ vertices and $b$ edges. Let $A^{(M)}\in \{0,1\}^{\binom{D}{M}\times b}$ be its $M^\text{th}$ incidence matrix. Then $(V,E)$ is an $M$-design if and only if $\vec{x}=\frac{1}{b}\vec{1}_b$ is a solution to the equation
\begin{equation}\label{eq:Appendix system of equations general M, 2}
    A^{(M)}\vec{x} = \frac{\binom{N}{M}}{\binom{D}{M}}\vec{1}_{\binom{D}{M}}.
\end{equation}
\end{theorem}

\begin{proof}

If the hypergraph is an $M$-design, then all $M$-subsets of $V$ appear in the same number of edges, which we call $\lambda$. In terms of $A^{(M)}$, this means that in each row of $A^{(M)}$, all the entries sum to $\lambda$, or equivalently $A^{(M)}\vec{1}_b = \lambda\vec{1}_{\binom{D}{M}}$. Plugging $\vec{x}=\frac{1}{b}\vec{1}_b$ into Eq.~(\ref{eq:Appendix system of equations general M, 2}), the LHS becomes $\frac{1}{b}A^{(M)}\vec{1}_b= \frac{\lambda}{b}\vec{1}_{\binom{D}{M}}$. Meanwhile, for an $M$-design, $\lambda$ is related to $D,N,$ and $b$ by~\cite{Stinson2004, ColbournDinitz2006}
\begin{equation}\label{eq:M-design parameter relation}
    \lambda = b\binom{N}{M}/\binom{D}{M}.
\end{equation}
Using Eq.~(\ref{eq:M-design parameter relation}), the LHS of Eq.~(\ref{eq:Appendix system of equations general M, 2}) can be further written as $\binom{N}{M}/\binom{D}{M}\vec{1}_{\binom{D}{M}}=\text{RHS}$, so $\vec{x}=\frac{1}{b}\vec{x}_b$ is indeed a solution.

Conversely, if $\vec{x}=\frac{1}{b}\vec{1}_b$ is a solution to Eq.~(\ref{eq:Appendix system of equations general M, 2}), then the equation reads $A^{(M)}\vec{1}_b=b\binom{N}{M}/\binom{D}{M}\vec{1}_{\binom{D}{M}}$. This implies that in each row of $A^{(M)}$, all the entries sum to exactly $b\binom{N}{M}/\binom{D}{M} = \lambda$, or equivalently, the hypergraph is an $M$-design where all $M$-subsets of $V$ appear in exactly $\lambda$ edges.

\end{proof}

\section{Nesting property of a maximally entangled state}\label{sec:appendix proof of theorem M leads to M'}

Here we prove the \hyperref[nesting-theorem]{nesting property} that if an $N$-fermion state is maximally $M$-body entangled, it is also $M'$-body entangled for any $M'\leq M$. We show two proofs of the claim: the first relies on Eq.~(\ref{eq:Diagonal Elements of MRDM in terms of Coefficients w.r.t. SDs}), and the second employs Eq.~(\ref{eq:System of equations general M}) and the language of hypergraphs.

\begin{theorem}\label{theorem:prop to identity gives propt to identity} 
Let $\ket{\Psi}$ be an $N$-fermion state with single-particle dimension $D$. Then $\forall M'\leq M$,
\begin{equation}
    \rho_\Psi^{(M)}=\frac{\binom{N}{M}}{\binom{D}{M}}\mathbb{I}_{\binom{D}{M}\times\binom{D}{M}}\Rightarrow \rho_\Psi^{(M')}=\frac{\binom{N}{M'}}{\binom{D}{M'}}\mathbb{I}_{\binom{D}{M'}\times\binom{D}{M'}}.
\end{equation}

\end{theorem}

\begin{proof}
To prove the claim, it suffices to show that $\rho_\Psi^{(M+1)}\propto\mathbb{I}\Rightarrow \rho_\Psi^{(M)}\propto\mathbb{I}$.

We consider an $N$-fermion state $\ket{\Psi}$ with $\rho_\Psi^{(M+1)}\propto\mathbb{I}$, and write it in terms of $N$-fermion SDs:
\begin{equation}
    \ket{\Psi}=\sum_{\beta}\gamma_\beta\ket{\beta},
\end{equation}
where $\gamma_\beta\neq 0$, $\ket{\beta}=C^{(N)\dagger}_\beta\ket{0}$, and $\beta$ is an $N$-subset of single-particle orbitals $[D]=\{1,\dots,D\}$. Since $ \rho_\Psi^{(M+1)}\propto\mathbb{I}$, the set of constituent SDs $\{\ket{\beta}\}$ of $\ket{\Psi}$ satisfies the \hyperref[overlap-criterion]{overlap criterion} that any two SDs in the set have less than $N-(M+1)$ overlapping orbitals. Consequently, the number of overlapping orbitals is less than $N-M$, and $\rho^{(M)}_\Psi$ is diagonal.

To see that $\rho_\Psi^{(M)}$ is proportional to the identity, we consider an $M$-subset of orbitals, $\{i_1,\dots,i_M\}$. Using Eq.~(\ref{eq:Diagonal Elements of MRDM in terms of Coefficients w.r.t. SDs}), we write the following $D-M$ diagonal entries of $\rho^{(M+1)}_\Psi$ in terms of the coefficients $\gamma_\beta$:
\begin{equation}\label{eq:Diagonal entries of rho^(M+1)}
    \rho_{\Psi,\{i_1,\dots,i_M,j\}\{i_1,\dots,i_M,j\}}^{(M+1)}=\sum_{\substack{\beta,\\ \beta \supseteq  \{i_1,\dots,i_M,j\}}}|\gamma_\beta|^2,
\end{equation}
where $j\in[D]\setminus  \{i_1,\dots,i_M\}$. For each $j$, the RHS of Eq.~(\ref{eq:Diagonal entries of rho^(M+1)}) is the sum of $|\gamma_\beta|^2$ over a set of SDs
\begin{equation}
    S_j=\{\ket{\beta}:\beta \supseteq  \{i_1,\dots,i_M,j\}\}.
\end{equation}
We claim that the set of all contributing SDs on the RHS of Eq.~(\ref{eq:Diagonal entries of rho^(M+1)}) for all values of $j$ in $[D]\setminus  \{i_1,\dots,i_M\}$ is the set of $N$-fermion SDs that contains $\{i_1,\dots,i_M\}$, i.e.,
\begin{equation}\label{eq:contributing SDs}
    \bigcup_j S_j = \{\ket{\beta}:\beta\supseteq  \{i_1,\dots,i_M\}\}.
\end{equation}
To see this, we notice that any $\ket{\beta}$ containing $ \{i_1,\dots,i_M,j\}$ for some $j\in[D]\setminus  \{i_1,\dots,i_M\}$ automatically contains $  \{i_1,\dots,i_M\}$. This means that in Eq.~(\ref{eq:contributing SDs}), $\text{LHS}\subseteq\text{RHS}$. On the other hand, if an $N$-fermion SD $\ket
{\beta}$ contains $M$ orbitals $  \{i_1,\dots,i_M\}$, where $M<N$, then it will also contain $M+1$ orbitals $\{i_1,\dots,i_M,j\}$ for some orbital $j\in[D]\setminus  \{i_1,\dots,i_M\}$. This implies that $\text{LHS}\supseteq\text{RHS}$ in Eq.~(\ref{eq:contributing SDs}). Therefore, Eq.~(\ref{eq:contributing SDs}) holds as claimed.

For a given $N$-fermion SD $\ket{\beta}$ containing $  \{i_1,\dots,i_M\}$, we consider the $j$ values in $[D]\setminus  \{i_1,\dots,i_M\}$ for which $\ket{\beta}$ will contribute a term $|\gamma_\beta
|^2$ on the RHS of Eq.~(\ref{eq:Diagonal entries of rho^(M+1)}). This happens for $j$ values when $ \{i_1,\dots,i_M,j\}\subseteq \beta$, or equivalently $j\in \beta\setminus  \{i_1,\dots,i_M\}$ due to $\{i_1,\dots,i_M\}$ being already contained in $\beta$. As a result, there are $|\beta\setminus \{i_1,\dots,i_M\}|=N-M$ values of $j$ for which the RHS of Eq.~(\ref{eq:Diagonal entries of rho^(M+1)}) has the term $|\gamma_\beta|^2$.

Now, we sum Eq.~(\ref{eq:Diagonal entries of rho^(M+1)}) over $j\in[D]\setminus  \{i_1,\dots,i_M\}$:
\begin{equation}\label{eq:sum of rho^(M+1) diagonal entries}
    \sum_j \rho_{\Psi, \{i_1,\dots,i_M,j\} \{i_1,\dots,i_M,j\}}^{(M+1)}
    =\sum_j \sum_{\substack{\beta,\\ \beta \supseteq  \{i_1,\dots,i_M,j\}}}|\gamma_\beta|^2.
\end{equation}
On the LHS, since $\rho^{(M+1)}_\Psi\propto\mathbb{I}$, all its diagonal entries have the same value $\binom{N}{M+1}/\binom{D}{M+1}$, and the LHS evaluates to $(D-N)\binom{N}{M+1}/\binom{D}{M+1}$. On the RHS, we showed in Eq.~(\ref{eq:contributing SDs}) that the $|\gamma_\beta|^2$ terms arise from the set of SDs $\{\ket{\beta}:\beta\supseteq  \{i_1,\dots,i_M\}\}$. Each SD $\ket{\beta'}$ in the set contributes $|\beta'\setminus  \{i_1,\dots,i_M\}|=N-M$ identical terms of $|\gamma_{\beta'}|^2$, as argued in the previous paragraph. Therefore, Eq.~(\ref{eq:sum of rho^(M+1) diagonal entries}) reads
\begin{equation}
    (D-N)\frac{\binom{N}{M+1}}{\binom{D}{M+1}}=(N-M)\sum_{\substack{\beta,\\ \beta\supseteq   \{i_1,\dots,i_M\}}}|\gamma_\beta|^2,
\end{equation}
or equivalently
\begin{equation}\label{eq:rho^M diagonal entry appendix B}
    \sum_{\substack{\beta,\\ \beta\supseteq   \{i_1,\dots,i_M\}}}|\gamma_\beta|^2 = \frac{D-M}{N-M}\frac{\binom{N}{M+1}}{\binom{D}{M+1}}=\frac{\binom{N}{M}}{\binom{D}{M}}.
\end{equation}
Here the LHS is $\rho^{(M)}_{\Psi,  \{i_1,\dots,i_M\}  \{i_1,\dots,i_M\}}$ by Eq.~(\ref{eq:Diagonal Elements of MRDM in terms of Coefficients w.r.t. SDs}). Since Eq.~(\ref{eq:rho^M diagonal entry appendix B}) holds for any $M$-subset $  \{i_1,\dots,i_M\}$, all the diagonal entries of $\rho^{(M)}_\Psi$ have the same value $\binom{N}{M}/\binom{D}{M}$, and we conclude
\begin{equation}
    \rho^{(M)}_\Psi=\frac{\binom{N}{M}}{\binom{D}{M}}\mathbb{I}.
\end{equation}

\end{proof}

Next, we present an alternative proof of Theorem \ref{theorem:prop to identity gives propt to identity}. Given a maximally $M$-body entangled $N$-fermion $\ket{\Psi}$, its associated hypergraph (i.e., the set of constituent SDs) satisfies the \hyperref[overlap-criterion]{overlap criterion} that any two edges share fewer than $N-M'$ vertices for any $M'\leq M$. We now prove that Eq.~(\ref{eq:System of equations general M}) admits a solution for any $M'\leq M$ as well. Together, these results imply that $\ket{\Psi}$ is maximally $M'$-body entangled for all $M'\leq M$.

\begin{theorem}\label{theorem:prop to identity gives propt to identity SOE}
Let $(V,E)$ be an N-uniform hypergraph with $D$ vertices and $b$ edges. Let $A^{(M)}\in \{0,1\}^{\binom{D}{M}\times b}$ be its $M^\text{th}$ incidence matrix, and $\vec{x}\in \mathbb{R}^b$. Then $\forall M'\leq M$
\begin{equation}
    A^{(M)}\vec{x} = \frac{\binom{N}{M}}{\binom{D}{M}}\vec{1}_{\binom{D}{M}}
    \text{   }\Rightarrow \text{   }
     A^{(M')}\vec{x} = \frac{\binom{N}{M'}}{\binom{D}{M'}}\vec{1}_{\binom{D}{M'}}.
\end{equation}
\end{theorem}

\begin{proof}
To prove the claim, it suffices to show that if the equation holds for $M$, then it holds for $M-1$. 

We first write the $M^\text{th}$ incidence matrix in terms of the incidence matrix. Each row of $A^{(M)}$ can be labeled by an $M$-subset $\{i_1,\dots,i_M\}$ of $\{1,\dots,D\}$, corresponding to an $M$-subset of $V$. Each column of $A^{(M)}$ can be labeled by an index $\alpha\in\{1,\dots,b\}$, corresponding to an edge in $E$. Consequently, each entry of $A^{(M)}$ can be labeled by a pair $(\{i_1,\dots,i_M\},\alpha)$. The entries of the coefficient matrix $A^{(M)}$ can be written in terms of the incidence matrix $A$ as
\begin{equation}\label{eq:A_c in terms of A}
    A^{(M)}_{\{i_1,i_2,\dots,i_M\}\alpha} = A_{i_1\alpha}A_{i_2\alpha}\cdots A_{i_M\alpha},
\end{equation}
where on the RHS, the $\alpha$ is \textit{not} summed over. 

The $\{i_1,\dots,i_M\}$-th row of the equation
$A^{(M)}\vec{x} = \binom{N}{M}/\binom{D}{M}\vec{1}_{\binom{D}{M}}$ then reads
\begin{equation}\label{eq:i_1 not in i_2 to i_M}
    \sum_{\alpha=1}^b A^{(M)}_{\{i_1,\dots,i_M\}\alpha} x_\alpha =    \sum_{\alpha=1}^b  A_{i_1\alpha}\cdots A_{i_M\alpha} x_\alpha = \frac{\binom{N}{M}}{\binom{D}{M}}.
\end{equation}
Note that here $i_1\notin\{i_2,\dots,i_M\}$. On the other hand, for $j\in\{i_2,\dots,i_M\}$,
\begin{equation}\label{eq:i_1 in i_2 to i_M}
    \sum_{\alpha=1}^b  A_{j\alpha}A_{i_2\alpha}\cdots A_{i_M\alpha} x_\alpha = \sum_{\alpha=1}^b A_{i_2\alpha}\cdots A_{i_M\alpha} x_\alpha,
\end{equation}
where we use the fact that $A_{j\alpha}^2=A_{j\alpha}$ due to $A_{j\alpha}\in\{0,1\}$. 

We define a vector $\vec{y}^{\{i_2,\dots,i_M\}}\in \mathbb{R}^b$ as
\begin{equation}
    y_\alpha^{\{i_2,\dots,i_M\}} \equiv A_{i_2\alpha}\cdots A_{i_M\alpha} x_\alpha,\text{ }\alpha\in\{1,\dots,b\},
\end{equation}
and denote the sum of its entries as $s$:
\begin{equation}
    s \equiv \sum_{\alpha=1}^b y_\alpha^{\{i_2,\dots,i_M\}}.
\end{equation}
In terms of $\vec{y}^{\{i_2,\dots,i_M\}}$, Eqs. (\ref{eq:i_1 not in i_2 to i_M}) and (\ref{eq:i_1 in i_2 to i_M}) read
\begin{equation}
     (A\vec{y}^{\{i_2,\dots,i_M\}})_j = 
\begin{cases}
s, & \text{if } j \in \{i_2,\dots,i_M\}, \\
\frac{\binom{N}{M}}{\binom{D}{M}}, & \text{otherwise. }
\end{cases}
\end{equation}
In matrix notation, this can be written as
\begin{equation}\label{eq:Ay in matrix notation}
    A\vec{y}^{\{i_2,\dots,i_M\}} = \frac{\binom{N}{M}}{\binom{D}{M}}\vec{1}_D + \left[s - \frac{\binom{N}{M}}{\binom{D}{M}}\right]\vec{1}_D^{\{i_2,\dots,i_M\}},
\end{equation}
where the vector $\vec{1}_D^{\{i_2,\dots,i_M\}}$ is defined as
\begin{equation}
     \vec{1}_{D,j}^{\{i_2,\dots,i_M\}} = 
\begin{cases}
1, & \text{if } j \in \{i_2,\dots,i_M\}, \\
0, & \text{otherwise. }
\end{cases}
\end{equation}

Now, we multiply both sides of Eq.~(\ref{eq:Ay in matrix notation}) by the $D\times D$ all-ones matrix $J_{D\times D}$:
\begin{equation}\label{eq:Ay in matrix notation, multiplied by J}
\begin{aligned}
    J_{D\times D}A\vec{y}^{\{i_2,\dots,i_M\}} &= \frac{\binom{N}{M}}{\binom{D}{M}}J_{D\times D}\vec{1}_D \\
    &+\left[s - \frac{\binom{N}{M}}{\binom{D}{M}}\right]J_{D\times D}\vec{1}_D^{\{i_2,\dots,i_M\}}.
\end{aligned}
\end{equation}
Here
\begin{align}
    J_{D\times D}A &= NJ_{D\times b},\\
    J_{D\times D}\vec{1}_D &= D\vec{1}_D,\\ 
    J_{D\times D}\vec{1}_D^{\{i_2,\dots,i_M\}} &= (M-1)\vec{1}_D.
\end{align}
Together with the fact that
\begin{equation}
    J_{D\times b}\vec{y}^{\{i_2,\dots,i_M\}} = \left(\sum_{\alpha=1}^b y_\alpha^{\{i_2,\dots,i_M\}}\right)\vec{1}_D = s\vec{1}_D,
\end{equation}
we can write Eq.~(\ref{eq:Ay in matrix notation, multiplied by J}) as
\begin{equation}
\begin{aligned}
    Ns\vec{1}_D = \frac{\binom{N}{M}}{\binom{D}{M}}D\vec{1}_D + \left[s - \frac{\binom{N}{M}}{\binom{D}{M}}\right](M-1)\vec{1}_D.
\end{aligned}
\end{equation}
Solving for $s$ we get
\begin{equation}
\begin{aligned}
   s = \frac{\binom{N}{M}}{\binom{D}{M}}\frac{D-M+1}{N-M+1} = \frac{\binom{N}{M-1}}{\binom{D}{M-1}},
\end{aligned}
\end{equation}
which leads to 
\begin{equation}
\begin{aligned}
  \sum_{\alpha=1}^b A_{i_2\alpha}\cdots A_{i_M\alpha} x_\alpha &= \frac{\binom{N}{M-1}}{\binom{D}{M-1}}.\label{eq:M-1 equation}
\end{aligned}
\end{equation}
The LHS of Eq.~(\ref{eq:M-1 equation}) is the $\{i_2,\dots,i_M\}$-th entry of the vector $A^{(M-1)} \vec{x}$. Since Eq.~(\ref{eq:M-1 equation}) holds for all $(M-1)$-subsets $\{i_2,\dots,i_M\}$, it tells us that all the entries of $A^{(M-1)} \vec{x}$ take the same value $\binom{N}{M-1}/\binom{D}{M-1}$. In other words,
\begin{equation}
    A^{(M-1)}\vec{x} = \frac{\binom{N}{M-1}}{\binom{D}{M-1}}\vec{1}_{\binom{D}{M-1}},
\end{equation}
as claimed.
\end{proof}

\section{Particle-hole symmetry}\label{sec:appendix particle-hole symmetry}

In this section, we show that for a $D$-dimensional single-particle Hilbert space, the presence (absence) of a maximally $M$-body entangled $N$-fermion state implies the presence (absence) of a maximally $M$-body entangled $(D-N)$-fermion state. Given a maximally $M$-body entangled $N$-fermion state, a maximally $M$-body entangled $(D-N)$-fermion state is constructed using the complement hypergraph of the hypergraph associated with the $N$-fermion state, as defined below.

For a hypergraph $(V,E)$, its \textit{complement hypergraph} is a hypergraph $(V,\olsi{E})$, where the vertex set remains $V$, and the edge set is given by $\olsi{E}=\{\olsi{e}\equiv V\setminus e:e\in E\}$. The incidence matrix $\olsi{A}$ of $(V,\olsi{E})$ has the same dimensions as the incidence matrix $A$ of $(V,E)$ and is given by $\olsi{A}=J-A$, where $J$ is the all-ones matrix. In the case where $(V,E)$ is $N$-uniform, i.e., all the edges in $E$ contain $N$ vertices, all the complement edges in $\olsi{E}$ contain $D-N$ vertices, making the complement hypergraph $(D-N)$-uniform.

For an $N$-fermion state $\ket{\Psi}$ with $\rho^{(M)}_\Psi\propto\mathbb{I}$, the associated hypergraph $(V,E)$ satisfies the \hyperref[overlap-criterion]{overlap criterion} that any two edges share fewer than $N-M$ vertices. We now prove that the complement hypergraph $(V,\olsi{E})$ also satisfies the overlap criterion, with any two edges sharing fewer than $(D-N)-M$ vertices.

\begin{theorem}[\textbf{Overlaps in the Complement Hypergraph}]\label{theorem:overlaps in the complement hypergraph}
Let $(V,E)$ be an $N$-uniform hypergraph with $D$ vertices and $p_{ij}=|e_i\cap e_j|$ be the number of common vertices between two edges $e_i,e_j\in E$. The number of common vertices $\olsi{p}_{ij}=|\olsi{e}_i\cap \olsi{e}_j|$ between two edges $\olsi{e}_i,\olsi{e}_j\in \olsi{E}$ of the complement hypergraph $(V,\olsi{E})$ is $\olsi{p}_{ij} = D - 2N + p_{ij}$. Additionally, if $p_{ij} < N - M \,\,$, then $\olsi{p}_{ij} < (D - N) - M$.
\end{theorem}
\begin{proof}
We can write the intersection of $\olsi{e}_i$ and $\olsi{e}_j$ as
\begin{equation}\label{eq:intersection of complement edges}
    \olsi{e}_i \cap \olsi{e}_j = (V\setminus e_i) \cap (V\setminus e_j) = V \setminus~(e_i \cup e_j).
\end{equation}
Since $e_i$ and $e_j$ both have $N$ elements and their intersection has size $p_{ij}$, the size of their union is $2N-p_{ij}$. The Eq.~(\ref{eq:intersection of complement edges}) then implies $\olsi{p}_{ij}=D-2N + p_{ij}$. If additionally $p_{ij}< N - M$, then $\olsi{p}_{ij} < N - M + D - 2N = (D - N) - M$.
\end{proof}

Next, we prove that the complement hypergraph $(V,\olsi{E})$ not only satisfies the \hyperref[overlap-criterion]{overlap criterion}, but also has a solution to its Eq.~(\ref{eq:System of equations general M}). Specifically, if $\vec{x}$ is a solution to Eq.~(\ref{eq:System of equations general M}) of $(V,E)$, then it is also a solution for $(V,\olsi{E})$. As a result, given a maximally $M$-body entangled $N$-fermion state $\ket{\Psi}=\sum_{k=1}^b\gamma_k \ket{\text{SD}_k}$ with associated hypergraph $(V,E)$, the $(D-N)$-fermion state $\ket{\ols{\Psi}}=\sum_{k=1}^b\gamma_k \ket{\overline{\text{SD}}_k}$ is also maximally $M$-body entangled. Here, the associated hypergraph of $\{\ket{\overline{\text{SD}}_k\}}_{k=1}^b$ is the complement hypergraph $(V,\olsi{E})$.

\begin{theorem}[\textbf{Particle-Hole Symmetry}]\label{theorem:particle-hole symmetry}
Let $(V,E)$ be an N-uniform hypergraph with $D$ vertices and $b$ edges. Let $A^{(M)}\in \{0,1\}^{\binom{D}{M}\times b}$ be its $M^\text{th}$ incidence matrix and $\vec{x}\in \mathbb{R}^b$. Then
\begin{equation}
    A^{(M)}\vec{x} = \frac{\binom{N}{M}}{\binom{D}{M}}\vec{1}_{\binom{D}{M}}
    \text{   }\Rightarrow \text{   }
    \olsi{A}^{(M)}\vec{x} = \frac{\binom{D-N}{M}}{\binom{D}{M}}\vec{1}_{\binom{D}{M}},
\end{equation}
where $\olsi{A}^{(M)}$ is the $M^\text{th}$ incidence matrix of the complement hypergraph $(V,\olsi{E})$.
\end{theorem}

\begin{proof}
Let $A\in\{0,1\}^{D\times b}$ be the incidence matrix of $(V,E)$. For the complement hypergraph $(V,\olsi{E})$ with incidence matrix $\olsi{A}=J-A$, we can write
\begin{equation}\label{eq:A_c entry for complement hypergraph}
\begin{aligned}
    \olsi{A}^{(M)}_{\{i_1,i_2,\dots,i_M\}\alpha} &= \olsi{A}_{i_1\alpha}\olsi{A}_{i_2\alpha}\cdots \olsi{A}_{i_M\alpha}\\
    &= (1- A_{i_1\alpha})(1-A_{i_2\alpha})\cdots(1-A_{i_M\alpha}) \\
    &= \prod_{k=1}^M(1-A_{i_k\alpha}).
\end{aligned}
\end{equation}
Each factor in this product is of the form $(1-A_{\bullet \alpha})$. By expanding this product as a sum, a term with order $k$ in $A_{\bullet \alpha}$ is proportional to
\begin{equation}\label{eq:k-th order term in A}
\underbrace{A_{\bullet \alpha}A_{\bullet \alpha}\cdots A_{\bullet \alpha}}_{k\text{ factors}},
\end{equation}
where the $k$ $\bullet$'s denote $k$ different indices. From Eq.~(\ref{eq:A_c in terms of A}), we identify Eq.~(\ref{eq:k-th order term in A}) as an entry of the $\alpha$-th column of the $k^\text{th}$ incidence matrix $A^{(k)}$ of $(V,E)$. With this in mind, we write the RHS of Eq.~(\ref{eq:A_c entry for complement hypergraph}) as
\begin{equation}\label{eq:binomial expansion}
\begin{aligned}
    \prod_{k=1}^M(1-A_{i_k\alpha}) &=
    (1-A_{\bullet\alpha})^M\\
    &= \sum_{k=0}^M\binom{M}{k}(-1)^kA_{\bullet\alpha}^k,
\end{aligned}
\end{equation}
where $A_{\bullet\alpha}^k$ has the meaning of Eq.~(\ref{eq:k-th order term in A}).

Using Eqs.~(\ref{eq:A_c entry for complement hypergraph}) and (\ref{eq:binomial expansion}), the $\{i_1,\dots,i_M\}$-th entry of $\olsi{A}^{(M)}\vec{x}$ can be written as 
\begin{equation}\label{eq:Entry of A_c x for complment}
\begin{aligned}
    \left[\olsi{A}^{(M)}\vec{x}\right]_{\{i_1,\dots, i_M\}} &= \sum_{\alpha=1}^b \olsi{A}_{\{i_1,\dots,i_M\}\alpha}x_\alpha\\
    &= \sum_{k=0}^M\binom{M}{k}(-1)^k\sum_{\alpha=1}^b A_{\bullet\alpha}^k x_\alpha.
\end{aligned}
\end{equation}
Here $\sum_{\alpha=1}^{b}A_{\bullet\alpha}^k x_\alpha$ is an entry of the vector $A^{(k)}\vec{x}$. We have shown in Theorem \ref{theorem:prop to identity gives propt to identity SOE} that $A^{(k)}\vec{x} = \binom{N}{k}/\binom{D}{k}\vec{1}_{\binom{D}{k}}$ for $k \leq M$, so all entries of $A^{(k)}\vec{x}$, hence $\sum_{\alpha=1}^{b}A_{\bullet\alpha}^k x_\alpha$, are equal to $\binom{N}{k}/\binom{D}{k}$. The Eq.~(\ref{eq:Entry of A_c x for complment}) then reads
\begin{equation}\label{eq:Entry of A_c x for Complement, 2nd form}
    \left[\olsi{A}^{(M)}\vec{x}\right]_{\{i_1,\dots, i_M\}} = \sum_{k=0}^M(-1)^k\binom{M}{k}\frac{\binom{N}{k}}{\binom{D}{k}}.
\end{equation}
To proceed, for a given integer $n$, we denote its $k$-th rising factorial as $n^{(k)}$ and the $k$-th falling factorial as $n_{(k)}$:
\begin{align}
    n^{(k)} &= \underbrace{n(n+1)\cdots(n+k-1)}_{k\text{ factors}},\\
    n_{(k)} &= \underbrace{n(n-1)\cdots(n-k+1)}_{k\text{ factors}}.
\end{align}
Using the fact that $(-n)^{(k)}=(-1)^k n_{(k)}$, we can write
\begin{equation}
    \frac{\binom{N}{k}}{\binom{D}{k}} = \frac{N_{(k)}}{D_{(k)}} 
    = \frac{(-N)^{(k)}}{(-D)^{(k)}}.
\end{equation}
Plugging this into Eq.~(\ref{eq:Entry of A_c x for Complement, 2nd form}), we obtain
\begin{equation}
\begin{aligned}
    \left[\olsi{A}^{(M)}\vec{x}\right]_{\{i_1,\dots,i_M\}} = \sum_{k=0}^M(-1)^k\binom{M}{k} \frac{(-N)^{(k)}}{(-D)^{(k)}} 1^k.
\end{aligned}
\end{equation}
The RHS by definition is the hypergeometric function
\begin{equation}
\begin{aligned}
    {}_{2}{F}_1\left(-M, -N; -D; 1\right) &= \frac{[-D-(-N)]^{(M)}}{(-D)^{(M)}} \\
    &= \frac{(D-N)_{(M)}}{D_{(M)}} = \frac{\binom{D-N}{M}}{\binom{D}{M}}.
\end{aligned}
\end{equation}
Combining the two, we get $\left[\olsi{A}^{(M)}\vec{x}\right]_{\{i_1,\dots,i_M\}} = \binom{D-N}{M}/\binom{D}{M}$ for any $M$-subset $\{i_1,\dots,i_M\}$. Hence, we see that all entries of $\olsi{A}^{(M)}\vec{x}$ have the same value, which implies that
\begin{equation}
    \olsi{A}^{(M)}\vec{x} = \frac{\binom{D-N}{M}}{\binom{D}{M}}\vec{1}_{\binom{D}{M}}.
\end{equation}
\end{proof}

\bibliography{refs}

\begin{thebibliography}{65}
\bibitem [{BCS()}]{BCS}%
  \BibitemOpen
  \href@noop {} {}\bibinfo {note} {We remark in passing, that the Bardeen-Cooper-Schrieffer wavefunction, which clearly is a superposition of many Slater determinants of distinct particle number, can be transformed to a single Slater determinant by the combination of a partial particle-hole transformation and the unitary Bogoliubov rotation.}\BibitemShut {Stop}%
\end{thebibliography}

\end{document}